\documentclass[reprint,pra,aps,notitlepage,nofootinbib,superscriptaddress]{revtex4-2}
\usepackage{amsmath}
\usepackage{amssymb}
\usepackage{bm,braket}
\usepackage{float}
\usepackage[com pat=1.1.0]{tikz-feynman}
\usepackage{hyperref}
\usepackage{multirow}

\hypersetup{
     colorlinks   = true,
     citecolor    = blue
}
\def\be{\begin{equation}}
\def\ee{\end{equation}}
\def\bea{\begin{eqnarray}}
\def\eea{\end{eqnarray}}
\def\bes{\begin{eqnarray}}
\def\ees{\end{eqnarray}}
\def\bi{\begin{itemize}}
	\def\ei{\end{itemize}} % ------- Define Greek Lowercase --------
	
\usepackage{amsthm}

\newtheorem{thm}{Theorem}
\newtheorem{lemma}[thm]{Lemma}

\newtheorem{prop}[thm]{Proposition}

\theoremstyle{definition}

\DeclareMathOperator{\A}{\mathcal{A}}
\DeclareMathOperator{\B}{\mathcal{B}}

\DeclareMathOperator{\M}{\mathcal{M}}

\newcommand{\pn}[1]{\left( #1 \right)}

\newcommand{\mat}[1]{\ensuremath{\begin{pmatrix} #1 \end{pmatrix}}}
\newcommand{\bigbrace}[1]{\left\{ #1 \right\}}
\usepackage{tikz}
\usetikzlibrary{braids}
\newcommand\stikz[1]{\tikz[baseline=(X.base)]{#1}} % helps align pictures on the same line
\begin{document}
\title{Quantum Computing with Two-dimensional Conformal Field Theories}
	\author{Elias Kokkas}
	\email{ikokkas@vols.utk.edu}
\affiliation{Department of Physics and Astronomy, The University of Tennessee, Knoxville, TN 37996-1200, USA}
\author{Aaron Bagheri}
\email{bagheri@math.ucsb.edu}
\affiliation{Department of Mathematics, University of California, Santa Barbara, CA 93106, USA}
\author{Zhenghan Wang}
\email{zhenghwa@microsoft.com}
\affiliation{Microsoft Station Q and Department of Mathematics, University of California, Santa Barbara, CA 93106, USA}
\author{George Siopsis}
	\email{siopsis@tennessee.edu}
\affiliation{Department of Physics and Astronomy, The University of Tennessee, Knoxville, TN 37996-1200, USA}
\date{\today}
\begin{abstract}
Conformal field theories have been extremely useful in our quest to understand physical phenomena in many different branches of physics, starting from condensed matter all the way up to high energy. Here we discuss applications of two-dimensional conformal field theories to
%CFT  minimal models $\mathcal{M}(k+2,k+1)$ with $k \geq 2$ for
fault-tolerant quantum computation based on the coset $ SU(2)_1^{\otimes k} / SU(2)_{k}$. We calculate higher-dimensional braiding matrices by considering conformal blocks involving $2N$ anyons, and search for gapped states that can be built from these conformal blocks. We introduce a gapped wavefunction that generalizes the Moore-Read state which is based on the critical Ising model, and show that our generalization leads to universal quantum computing. 
\end{abstract}
\maketitle

\section{Introduction}
Topological quantum computation was first introduced \cite{bib:1,bib:2,bib:2a} and further developed \cite{bib:3,bib:4,bib:5} as an elegant approach to fault-tolerant quantum computation which utilizes certain quasi-particles called anyons. These are exotic quasi-particles that live in two spatial dimensions and exhibit quantum statistics that are neither fermionic nor bosonic. We can distinguish between two different types of anyons, Abelian and non-Abelian. Abelian anyons are associated with the one-dimensional representation of the braid group, and were first studied in Ref.\ \cite{bib:6}. Their quantum state acquires a global phase under the exchange of two identical particles, whereas the exchange of two identical non-Abelian anyons changes their quantum state via a unitary matrix. Another important difference between Abelian and non-Abelian anyons is with regard to their fusion rules. Abelian anyons fuse into a single Abelian anyon, whereas non-Abelian anyons have multiple fusion outcomes (channels). We can store and process information using the fusion rules and braiding statistics of these anyons, respectively. However, only non-Abelian anyons can be used to implement complicated quantum gates that are not proportional to the identity. Ising and Fibonacci anyons are the simplest candidates for a topological quantum computer \cite{bib:7,bib:8}, but only the latter offers a universal set of quantum gates via braiding. Nevertheless, Abelian anyons are still useful for quantum computing tasks, such as quantum memory. We achieve
fault tolerance by encoding information non-locally and processing it using braidings that depend only on the topology of anyons. Anyons emerge as localized excitations in a topological phase of matter provided there exists an energy gap and a ground state topological degeneracy which is robust against external interactions \cite{bib:9,bib:10}. Of course, errors associated with wrong braidings can still occur.

Despite the enormous theoretical success of anyons, their physical realization is to this day a challenge. Superconductor-semiconductor nanowires are promising candidates for Majorana zero modes \cite{bib:11,bib:12} which are quasi-particles that obey the same fusion and braiding rules as the Ising anyons. Another approach is by studying systems with the Fractional Quantum Hall Effect (FQHE). Experimental data \cite{bib:13} support the emergence of Abelian anyonic excitations at the quantum Hall effect for $\nu=\frac{1}{3}$. There is some evidence that Ising anyons can be found in the $\nu=\frac{5}{2}$ quantum Hall state and Fibonacci anyons exist in the $\nu=\frac{12}{5}$ quantum Hall state. However, experimental results are not conclusive.

To better understand the FQHE and its quantum statistics, one needs to understand the wavefunction of the ground state and its quasi-hole excitations. Na\"\i vely, one would have to determine a many-body  Hamiltonian that could be diagonalized to obtain its eigenstates. Alternatively, following the pioneering work of Moore and Read \cite{bib:14}, one can construct wavefunctions for these states using conformal blocks of certain conformal field theories (CFTs). For example, one can construct the Laughlin wavefunction  \cite{bib:15} that describes the FQHE at filling $\nu=\frac{1}{q}$, with $q$ an odd integer number, and supports Abelian anyons, using a CFT with central charge $c=1$ consisting of a free massless boson. Additionally, Moore and Read \cite{bib:14} constructed the Pfaffian wavefunction, which obeys non-Abelian statistics, using the critical Ising CFT minimal model $\mathcal{M}(4,3)$ to describe the $\nu=\frac{5}{2}$ FQHE. This spurred a lot of activity in the subject \cite{bib:16,bib:17,bib:18,bib:19,bib:20}. Similar proposals suggest that the $\nu=\frac{5}{2}$ state is described by the anti-Pfaffian state \cite{bib:21,bib:22}. 

In order to build a wavefunction describing a number of non-Abelian quasi-holes from CFT correlators, there are a few caveats to address. Since the conformal blocks for non-Abelian anyons are multi-valued functions, the position coordinates of the non-Abelian quasi-holes can only play the role of parameters in the wavefunction of the system. The role of coordinates is assumed by the position of Abelian particles which must be present in the system. There presence induces singularities which are removed by the inclusion of the Jastrow factor which describes an independent system defined by a different CFT. Moreover, as was emphasized in \cite{bib:18}, for the braiding statistics of the wavefunction to match the monodromy around the branch points of the multi-valued part of the function, we need to ensure that the Berry holonomy vanishes. Fortunately, it was later demonstrated in \cite{bib:20} using the plasma analogy that indeed the Berry holonomy vanishes for the Moore-Read (MR) wavefunction. Attempts to construct MR-like wavefunctions for minimal models $\mathcal{M}(m+1,m)$ with $m > 3$ were also studied in \cite{bib:20}, but it was realized that they cannot describe gapped states because the plasma is not screened.

In this work, we propose an alternative generalization of the MR wavefunction based on the coset $SU(2)_1^{\otimes k} /SU(2)_k $, where $SU(2)_k$ is the Wess-Zumino-Witten (WZW) model based on the gauge group $SU(2)$ at level $k$. For $k=2$, it reduces to the MR wavefunction, because the critical Ising model $\mathcal{M}(4,3)$ is isomorphic to the coset $SU(2)_1\otimes SU(2)_1 /SU(2)_2 $ \cite{bib:23,bib:24,bib:25}. Unlike in the $k=2$ case, higher values of $k$ lead to universal quantum computing involving only braiding. Using the plasma analogy, we show that our wavefunction is gapped, ensuring fault-tolerant quantum computing. It should be pointed out that the minimal model $\mathcal{M} (k+2,k+1)$ can be constructed from a similar coset, $SU(2)_{k-1}\times SU(2)_1/SU(2)_{k}$. However, even though braiding alone can be shown to lead to universal quantum computation for $k>2$, no gapped state has been constructed. In our coset construction based on $SU(2)_1^{\otimes k} /SU(2)_k $, there is a primary field of conformal dimension $\frac{1}{2}$ for all $k$, leading to a gapped state. It would be interesting to identify a physical realization of our theoretical construct.

Our discussion is organized as follows. In Section  \ref{sec:II}, we review the pertinent features of Virasoro minimal models using the Coulomb gas formalism. In Section \ref{sec:III}, we calculate the braiding and fusion matrices for four-point and six-point amplitudes. In Section \ref{sec:coamp}, we calculate amplitudes in the coset CFT. In Section \ref{sec:IV}, we discuss braiding from the point of view of anyon models. In Section \ref{sec:V}, we construct a wavefunction that generalizes the MR wavefunction and leads to universal fault-tolerant quantum computing. Finally, in Section \ref{sec:con} we present our conclusions. Details of our calculations can be found in Appendices \ref{app:A} (exchange matrices) and \ref{app:B} (amplitudes in the $SU(2)_q$ WZW model).

\section{Virasoro minimal models} \label{sec:II}

In this Section, we review the salient features of the Virasoro minimal models. The minimal model $\mathcal{M} (k+2,k+1)$ shares common features with the coset CFT $SU(2)_1^{\otimes k} /SU(2)_k $ that we are interested in, such as a set of primary fields ($\Phi_{(1,s)}$ with $1 \leq s \leq k+1$). These fields have the same conformal dimensions, fusion rules, and braiding statistics in both CFTs. For $k=2$, the coset CFT coincides with the critical Ising model. For all $k$, the coset CFT contains a primary field $\psi$ of conformal dimension $h_\psi = \frac{1}{2}$, which is not present in minimal models, except for $k=2$. The absence of a similar field in minimal models with $k >2$ prevents us from using them for fault-tolerant universal quantum computation. As we will show in Section \ref{sec:V}, the field $\psi$ obeys Abelian fusion rules, which is crucial for the construction of a gapped wavefunction based on the coset CFT $SU(2)_1^{\otimes k} /SU(2)_k $.  

\begin{table}[b]
\caption{\label{table:1}%
Charge and dimension of primary fields in the minimal model $\mathcal{M}(k+2,k+1)$ and the coset CFT $SU(2)_1^{\otimes k} /SU(2)_k $.}
\begin{ruledtabular}
\begin{tabular}{lccc}
\textrm{Primary field}&
\textrm{Symbol}&
\textrm{Dimension}&
\textrm{Charge}\\
\colrule
$\Phi_{(1,1)}$ & $\mathbb{I}$ & 0 & 0 \\
$\Phi_{(1,2)}$ & $\sigma $& $\frac{k-1}{4(k+2)}$ & $\frac{k+1}{2\sqrt{(k+1)(k+2)}}$ \\
$\Phi_{(1,3)}$ & $\varepsilon$ & $\frac{k}{k+2}$ & $\frac{k+1}{\sqrt{(k+1)(k+2)}}$ \\
$\Phi_{(1,4)}$ & $\varepsilon'$  & $\frac{3(3k+1)}{4(k+2)}$ & $\frac{3(k+1)}{2\sqrt{(k+1)(k+2)}}$ \\
\vdots & & \vdots & \vdots \\
$\Phi_{(1,k+1)}$ & & $\frac{k(k-1)}{4}$ & $\frac{k(k+1)}{2\sqrt{(k+1)(k+2)}}$
\end{tabular}
\end{ruledtabular}
\end{table}

In the Coulomb gas formalism \cite{bib:26,bib:27,bib:28}, one starts with a massless scalar field $\varphi$ in two spacetime dimensions, and adds a background charge $\alpha_0$ at infinity, which shifts the central charge to $c=1-24\alpha_0^2$. In terms of the integer $k$ the background charge is given by
\be \label{eq:II-1} \alpha_0 = \frac{1}{2\sqrt{(k+1)(k+2)}} \ . \ee
Physical observables, such as spin and energy density, are represented by primary fields expressed as vertex operators $\Phi_\alpha(\eta,\bar{\eta})= e^{i \sqrt{2}\alpha \varphi (\eta,\bar{\eta})}$, including both a holomorphic and an anti-holomorphic part,
where $\alpha$ is the charge. The primary field $\Phi_\alpha$ has conformal dimension $h_\alpha=\alpha^2-2 \alpha_0 \alpha$. The corresponding observable may also be represented by the conjugate vertex operator $\tilde{\Phi}_\alpha \equiv \Phi_{2\alpha_0-\alpha}$, which has the same conformal dimension as $\Phi_\alpha$.

The minimal model $\mathcal{M} (k+2,k+1)$ possesses a finite number of primary fields labeled by a pair of integers $(r,s)$, where $r=1,\dots,k$, and $s=1,\dots,k+1$. The charge and conformal dimension of the primary field $\Phi_{(r,s)}$ are given, respectively, by
\be \label{eq:II-2}  \alpha_{(r,s)} =  \frac{(1-r)(k+2)-(1-s)(k+1)}{2\sqrt{(k+1)(k+2)}} \ , \ee
\be \label{eq:II-3}  h_{(r,s)} = \frac{[r(k+2) - s(k+1)]^2-1}{4(k+1)(k+2)} \ . \ee
The conjugate field is $\tilde\Phi_{(r,s)}=\Phi_{(k+1-r,k+2-s)}$.
%that has the same conformal dimension but different charge. The total number of distinct primary fields is $\frac{k(k+1)}{2}$ and the dimension of the Hilbert space depends on $k$ via the fusion rules.

Correlators of primary fields can be split into holomorphic and antiholomorphic parts by splitting the scalar field $\varphi (\eta,\bar{\eta}) = \varphi (\eta) + \bar{\varphi} (\bar{\eta})$. For quantum computing, we are interested in chiral CFTs. We will concentrate on the holomorphic part of correlators. Of particular importance are the primary fields with $r=1$. They are common features in the minimal model $\mathcal{M} (k+2,k+1)$ and the coset CFT $SU(2)_1^{\otimes k} / SU(2)_k$ and form a closed algebra thanks to the fusion rules 
\be\label{eq:II-4}
\Phi_{(1,s)} \otimes \Phi_{(1,s')} = \sum_{ s''\stackrel{2}{=} |s'-s|+1 }^{\text{min} (s+s'-1,2k+3-s-s')} \Phi_{(1,s'')} \ .
\ee
where $\stackrel{2}{=}$ denotes incrementing the summation variable by $2$.
More generally, the fusion rules are
\bea\label{eq:fusion}
	&& \Phi_{(r,s)} \otimes \Phi_{(r',s')} = \nonumber\\ &&\ \ \ \ \sum_{r'' \stackrel{2}{=} |r'-r|+1}^{\min(r+r'-1, 2q-1-r-r')} \sum_{s'' \stackrel{2}{=} |s'-s|+1}^{\min(s+s'-1, 2p-1-s-s')} \Phi_{(r'',s'')} \ . \nonumber\\
\eea
For $k \geq 2$, the charge and dimension of these primary fields are summarized in Table \ref{table:1}. We will build Hilbert spaces of qubits based on correlators of the holomorphic part of the primary field $\Phi_{(1,2)}$, which we denote by $\sigma (z)$. Other fields $\Phi_{(1,s)}$ contribute to correlators of $\sigma$ as intermediate states.
We denote their holomorphic part by $\varepsilon, \varepsilon', \dots$, for $s=3,4,\dots$, respectively. Conjugate fields are denoted by $\tilde{\sigma}, \tilde{\varepsilon}, \tilde{\varepsilon'}, \dots$, and their anti-holomorphic counterparts by $\bar{\sigma}, \bar{\varepsilon}, \bar{\varepsilon'}, \dots$. Thus, e.g., $\Phi_{(1,2)} (\eta,\bar{\eta}) = \sigma (\eta) \bar{\sigma} (\bar{\eta})$.

The anomalous $U(1)$ symmetry of the massless scalar field coupled to the background charge leads to a charge neutrality condition which states that the total charge in a correlator has to be equal to twice the background charge,
\be \label{eq:II-5}  \sum_i \alpha_i= 2\alpha_0 \ .\ee
To define non-vanishing correlators of physical observables that obey the neutrality condition \eqref{eq:II-5}, one introduces screening operators $Q_\pm$ of zero conformal dimension,
\be \label{eq:II-6} Q_\pm = \int d^2 w V_\pm (w) \bar{V}_{\pm} (\bar{w}) \ , \ \ V_\pm (w) = e^{i\sqrt{2} \alpha_\pm \varphi (w)} \ ,\ee
of charge $\alpha_+ = \sqrt{\frac{k+2}{k+1}}$  and $\alpha_- =- \sqrt{\frac{k+1}{k+2}}$, respectively.

The correlation function of $2N$ primary fields $\Phi_{(1,2)}$ is given by 
\be \label{eq:II-7}  G^{(2N)}(\bm{\eta}, \bar{\bm{\eta}} ) =  \left\langle \bm{\sigma}_1 \cdots \bm{\sigma}_{2N-1} \tilde{\bm{\sigma}}_{2N} Q_-^{N-1}  \right\rangle \ee
where $\bm{\eta} = (\eta_1,\dots, \eta_{2N})$, $\bm{\sigma} = \Phi_{(1,2)}$, and $\bm{\sigma}_j = \bm{\sigma} (\eta_j,\bar{\eta}_j)$. To define this correlator, we inserted $n-1$ screening operators $Q_-$ and used the conjugate field for one of the primary fields.

The non-chiral correlator \eqref{eq:II-7} can be split into holomorphic and antiholomorphic parts as \cite{bib:29}
\be \label{eq:II-8}  G^{(2N)}(\bm{\eta}, \bar{\bm{\eta}} ) = \sum_\mu | \mathcal{F}_\mu^{(2N)} (\bm{\eta} ) |^2 \ee
where we sum over the conformal blocks of the chiral model labeled by $\mu$. The chiral correlator (conformal block) is given by
\be \label{eq:II-9}  \mathcal{F}_\mu^{(2N)} (\bm{\eta} ) = \sqrt{N_\mu} \oint_\mu d^{N-1}\bm{w} \, \mathcal{I}^{(2N)} ( \bm{\eta}; \bm{w}) \ee 
where
\be \label{eq:II-10}  \mathcal{I}^{(2N)} ( \bm{\eta}; \bm{w}) = \left\langle \sigma_1  \cdots \sigma_{2N-1} \tilde{\sigma}_{2N} \prod_{j=1}^{N-1} V_- (w_j)  \right\rangle \ee
where $\bm{w} = (w_1, \dots,w_{N-1})$, $\sigma_j = \sigma(\eta_j)$, and $N_\mu$ are normalization constants determined by matching the expressions in Eqs.\ \eqref{eq:II-7} and \eqref{eq:II-8}.
%\color{red}{defined by}
%\be
%\int d^{2(N-1)} \bm{w}\, |\mathcal{I} ( \bm{z}; \bm{w}) |^2 = \sum_{\mu} N_{\mu} \Big| { \oint_\mu d^{N-1}\bm{w} \, \mathcal{I}^{(2N)} ( \bm{z}; \bm{w}) } \Big|^2
%\ee 
%\color{black}
%which is fixed (up to a phase) by monodromy transformations. 
The conformal block \eqref{eq:II-9} is obtained by performing $N-1$ contour integrals. One distinguishes between different conformal blocks by the position of the contours of integration; $\mu$ labels the collective choice.

%The non-chiral correlation function \eqref{eq:II-7}  can be written as an integral over the entire complex plane
%\be \label{eq:II-11}  G(\bm{z}, \bar{\bm{z}} ) = \int d^{2(N-1)} \bm{w}\, |\mathcal{I} ( \bm{z}; \bm{w}) |^2  \ .\ee 

\section{Braiding and Fusion Matrices}\label{sec:III}
As discussed in the previous section, chiral amplitudes are not single-valued functions since they depend on the choice of contours of integration. Conformal blocks form a basis for these amplitudes of dimensionality that depends on the number of primary fields and therefore on the integer $k$ labeling the CFT.  This basis is mapped onto the basis for the Hilbert space of qubits in quantum computation.

One can deduce the number of independent conformal blocks directly from the fusion rules of the model. The exchange of two primary fields $\sigma$ at positions $\eta_i$ and $\eta_j$ is equivalent to a change of basis from $\mathcal{F}_\mu$ to $\mathcal{F}'_\mu$ via an exchange matrix,
\be\label{eq:III-1} \mathcal{F}'_\mu =\sum_\nu (R_{ij})_{\mu \nu} \mathcal{F}_\nu \ .\ee
These exchange matrices result into braiding and fusion matrices \cite{bib:28,bib:30,bib:31} that can be mapped onto quantum gates. We discuss how this is done in detail for four- and six-point amplitudes.

\subsection{Four-point amplitudes}
The simplest nontrivial amplitude is the four-point correlation function involving $\sigma$ fields. There are two conformal blocks associated with this four-point function for all values of the integer $k$, as can be easily deduced from the fusion rules \eqref{eq:II-4}. They form a two-dimensional Hilbert space corresponding to a single qubit. The amplitude \eqref{eq:II-10} with $N=2$ can be simplified using conformal invariance which allows one to fix $\eta_1 \rightarrow 0$, $\eta_2 \rightarrow x$, $\eta_3 \rightarrow 1$, and $z_4 \rightarrow \infty$, where $x = \frac{\eta_{12} \eta_{34}}{\eta_{13} \eta_{24}}$ is the anharmonic ratio with $\eta_{ij} = \eta_i - \eta_j$. Omitting a factor $\left( \frac{\eta_{13} \eta_{24} }{\eta_{12} \eta_{23} \eta_{34} \eta_{41} } \right)^{2h_\sigma}$, we may write
\be  \label{eq:III-3} \mathcal{I}^{(4)} (x; w) =  \left\langle \sigma (0)  \sigma (x) \sigma (1) \tilde{\sigma} (\infty)  V_-  (w)  \right\rangle \ .\ee
The non-chiral four-point function can be written in terms of conformal blocks as
\be  \label{eq:III-6} G^{(4)} (x,\bar{x}) = |\mathcal{F}^{(4)}_1 (x)|^2 + |\mathcal{F}^{(4)}_2 (x)|^2 \ . \ee
To define the two conformal blocks we need to carefully choose the contour of integration in order to avoid the branch points and singularities at $0, x, 1, \infty$.  This can be done by choosing two branch cuts along the real axis, one that goes from $0$ to $z$ and another one from $1$ to $\infty$. We obtain two different contours which encircle $(0,x)$ and $(1,\infty)$, respectively. After shrinking these contours, the two conformal blocks are defined as
\bea  \label{eq:III-7} \mathcal{F}^{(4)}_1 (x) &=& \sqrt{N_1} [x (1-x)]^{\frac{k+1}{2(k+2)}} \nonumber\\
&& \times \int_0^x dw\,   [w (x-w) (1-w)]^{-\frac{k+1}{k+2}} \nonumber \\ 
\mathcal{F}^{(4)}_2 (x) &=&  \sqrt{N_2} [x (1-x)]^{\frac{k+1}{2(k+2)}}    \nonumber \\ 
&& \times \int_1^\infty dw \,  [w (w-x) (w-1)]^{-\frac{k+1}{k+2}} \ .\eea
After some algebra, we obtain these conformal blocks in terms of Hypergeometric functions 
\bea  \label{eq:III-8} \mathcal{F}^{(4)}_1 (x) &=& \sqrt{N_1} \frac{\Gamma^2 ( \frac{1}{k+2} )}{\Gamma ( \frac{2}{k+2})}  x^{\frac{1-k}{2(k+2)}} (1-x)^{\frac{k+1}{2(k+2)}} \nonumber\\
&&\times \;_2F_1 \left( \frac{k+1}{k+2}, \frac{1}{k+2}; \frac{2}{k+2}; x \right) \nonumber\\
\mathcal{F}^{(4)}_2 (z) &=& \sqrt{N_2} \frac{\Gamma ( \frac{1}{k+2} )\Gamma ( \frac{2k+1}{k+2} )}{\Gamma (\frac{2k+2}{k+2})} [x (1-x)]^{\frac{k+1}{2(k+2)}} \nonumber\\
&&\times \;_2F_1 \left( \frac{k+1}{k+2}, \frac{2k+1}{k+2}; \frac{2(k+1)}{k+2}; x \right)  \ .\eea
To understand the physical content of these conformal blocks, we consider the $x\to 0$ limit. We observe that  $\mathcal{F}^{(4)}_1 (x) \sim x^{-\frac{k-1}{2(k+2)}} [1 + \mathcal{O} (x)]$, whereas $\mathcal{F}^{(4)}_2 (x) \sim x^{\frac{k+1}{2(k+2)}} [1 + \mathcal{O} (x)]$. Comparing with the operator product expansion (OPE)
$ \sigma (x) \sigma(0) \sim x^{-\frac{k-1}{2(k+2)}} \mathbb{I} + x^{ \frac{k+1}{2(k+2)}} \varepsilon (0) $
where $\varepsilon$ is defined in Table \ref{table:1}, it is evident that $\mathcal{F}^{(4)}_1$ and $\mathcal{F}^{(4)}_2$ have intermediate state $\mathbb{I}$  and $\varepsilon$, respectively. Schematically, they are given by the two diagrams shown in Figure \ref{fig:III-1}.

\begin{figure}[ht!]
    \centering
    \begin{tikzpicture}
  \begin{feynman}
    \vertex at (0,0) (a);
    \vertex at (-1,-1) (f1) {\(\sigma_1\)};
    \vertex at (-1,1) (f2) {\(\sigma_2\)};
    \vertex at (1,0) (b);
    \vertex at (0.5,0.3) (c) {\(\mathbb{I}\)};
    \vertex at (2,1) (f3) {\(\sigma_3\)};
    \vertex at (2,-1) (f4) {\(\sigma_4\)};
    \vertex at (-2,0)  (d1) {\(\mathcal{F}^{(4)}_1\)};
    \vertex at (-1.5,0)  (d2) {\(=\)};

    \diagram* {
      (f1) --  (a) --  (f2),
      (a)  -- (b),
      (f3) -- (b) -- (f4)
    };
  \end{feynman}
\end{tikzpicture}

\begin{tikzpicture}
\begin{feynman}
    \vertex at (0,0) (a);
    \vertex at (-1,-1) (f1) {\(\sigma_1\)};
    \vertex at (-1,1) (f2) {\(\sigma_2\)};
    \vertex at (1,0) (b);
    \vertex at (0.5,0.3) (c) {\(\varepsilon\)};
    \vertex at (2,1) (f3) {\(\sigma_3\)};
    \vertex at (2,-1) (f4) {\(\sigma_4\)};
    \vertex at (-2,0)  (d1) {\(\mathcal{F}^{(4)}_2\)};
    \vertex at (-1.5,0)  (d2) {\(=\)};

    \diagram* {
      (f1) --  (a) --  (f2),
      (a)  -- (b),
      (f3) -- (b) -- (f4)
    };
\end{feynman}
\end{tikzpicture}
    \caption{Conformal blocks of the four-point function.}
    \label{fig:III-1}
\end{figure}
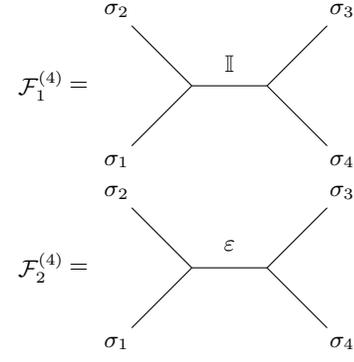

The normalization constants $N_\mu$ ($\mu =1,2$) are determined by comparing the expressions \eqref{eq:II-7} for $N=2$ and \eqref{eq:III-6} for the non-chiral amplitude $G^{(4)}$. Calculating the non-chiral amplitude can be avoided by using an argument based on monodromy transformations around $0$ and $1$. Under a monodromy transformation, we change bases. The conformal blocks in the new basis must provide a decomposition of the non-chiral amplitude of the same form \eqref{eq:III-6}. This leads to linear constraints that determine the normalization constants up to an overall multiplicative factor, which suffices for our application to quantum computation. 
%Using can be used to fix the normalization constants $N_\mu$, $\mu=1,2$. First we need to introduce the unnormalized conformal blocks $\mathcal{G}_{\mu}=N^{-\frac{1}{2}}_\mu \mathcal{F}_\mu$.  One can observe that  the monodromy of $\mathcal{G}_{\mu}(z)$ around $0$ is diagonal, justifying the definition of the four-point function in terms of the two conformal blocks as $G(\bm{z}, \bar{\bm{z}} ) = \sum_{\mu} N_\mu | \mathcal{G}_\mu (\bm{z} ) |^2 $. The monodromy of $\mathcal{G}_{\mu}(z)$ around $1$ is not diagonal in this basis, however we can change our basis via $\mathcal{G}_{\mu}(z)=\sum_{\nu} (\hat{R}_{13}^{-1} )_{\mu\nu} \mathcal{G}_{\nu}(1-z)$. Here the matrix $\hat{R}_{13}$ corresponds to a basis change that we can be obtained by exchanging the operators $z_1$ and $z_3$ on the unnormalized conformal blocks $\mathcal{J}_\mu$ and is in general a non-unitary matrix, it should not be confused with the unitary matrix ${R}_{13}$.   The four-point function becomes  $G(\bm{z}, \bar{\bm{z}} ) = \sum_{\mu \nu \lambda} N_\mu (\hat{R}_{13}^{-1})_{\mu \nu} (\hat{R}_{13}^{-1})_{\mu \lambda} \mathcal{G}_\nu (\bm{1-z} ) \mathcal{G}^*_\lambda (\bm{1-z} ) $ and since in this basis the monodromy around $1$ is diagonal we get the constraint
%\be  \label{eq:III-9} \sum_{\mu} N_{\mu} (\hat{R}_{13}^{-1})_{\mu \nu} (\hat{R}_{13}^{-1})_{\mu \lambda}=0 \;\;\;\;\;\; \forall \; \nu \neq \lambda \ . \ee
After some algebra, we obtain
\be
N_1 = \mathcal{N} \sin  \frac{\pi }{k+2} \ , \ \ N_2 = \mathcal{N} \sin \frac{3 \pi }{k+2}\ .
\ee
where $\mathcal{N}$ can be determined using \eqref{eq:II-7}, but is not needed for our purposes.

For the four-point chiral amplitude, we derive two braiding matrices, $R_{12}$ and $R_{23}$, and a fusion matrix $R_{13}$, where $R_{ij}$ corresponds to the exchange of positions $\eta_i \leftrightarrow \eta_j$. These matrices are defined diagrammatically in Figure \ref{fig:III-2}.  
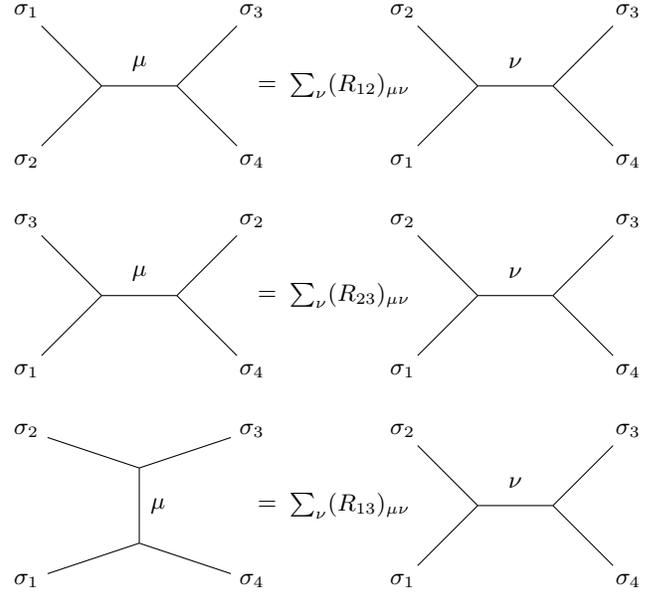
\begin{figure}
\centering
    \begin{tikzpicture}
    \begin{feynman}
    \vertex at (0,0) (a);
    \vertex at (-1,-1) (f1) {\(\sigma_2\)};
    \vertex at (-1,1) (f2) {\(\sigma_1\)};
    \vertex at (1,0) (b);
    \vertex at (0.5,0.3) (c) {\(\mu\)};
    \vertex at (2,1) (f3) {\(\sigma_3\)};
    \vertex at (2,-1) (f4) {\(\sigma_4\)};
    \vertex at (2.2,0) (d1) {\(=\)};
    \vertex at (3.3,0) (d2) {\(\sum_{\nu} (R_{12})_{\mu \nu}\)};
    \vertex at (5,0) (na);
    \vertex at (4,-1) (nf1) {\(\sigma_1\)};
    \vertex at (4,1) (nf2) {\(\sigma_2\)};
    \vertex at (6,0) (nb);
    \vertex at (5.5,0.3) (nc) {\(\nu\)};
    \vertex at (7,1) (nf3) {\(\sigma_3\)};
    \vertex at (7,-1) (nf4) {\(\sigma_4\)};

    \diagram* {
      (f1) --  (a) --  (f2),
      (a)  -- (b),
      (f3) -- (b) -- (f4),
      (nf1) --  (na) --  (nf2),
      (na)  -- (nb),
      (nf3) -- (nb) -- (nf4)
    };
    \end{feynman}
    \end{tikzpicture}
    $\;$\\
    \begin{tikzpicture}
    \begin{feynman}
    \vertex at (0,0) (a);
    \vertex at (-1,-1) (f1) {\(\sigma_1\)};
    \vertex at (-1,1) (f2) {\(\sigma_3\)};
    \vertex at (1,0) (b);
    \vertex at (0.5,0.3) (c) {\(\mu\)};
    \vertex at (2,1) (f3) {\(\sigma_2\)};
    \vertex at (2,-1) (f4) {\(\sigma_4\)};
    \vertex at (2.2,0) (d1) {\(=\)};
    \vertex at (3.3,0) (d2) {\(\sum_{\nu} (R_{23})_{\mu \nu}\)};
    \vertex at (5,0) (na);
    \vertex at (4,-1) (nf1) {\(\sigma_1\)};
    \vertex at (4,1) (nf2) {\(\sigma_2\)};
    \vertex at (6,0) (nb);
    \vertex at (5.5,0.3) (nc) {\(\nu\)};
    \vertex at (7,1) (nf3) {\(\sigma_3\)};
    \vertex at (7,-1) (nf4) {\(\sigma_4\)};

    \diagram* {
      (f1) --  (a) --  (f2),
      (a)  -- (b),
      (f3) -- (b) -- (f4),
      (nf1) --  (na) --  (nf2),
      (na)  -- (nb),
      (nf3) -- (nb) -- (nf4)
    };
    \end{feynman}
    \end{tikzpicture}
    $\;$\\
    \begin{tikzpicture}
    \begin{feynman}
    \vertex at (0.5,0.5) (a);
    \vertex at (-1,-1) (f1) {\(\sigma_1\)};
    \vertex at (-1,1) (f2) {\(\sigma_2\)};
    \vertex at (0.5,-0.5) (b);
    \vertex at (0.75,0) (c) {\(\mu\)};
    \vertex at (2,1) (f3) {\(\sigma_3\)};
    \vertex at (2,-1) (f4) {\(\sigma_4\)};
    \vertex at (2.2,0) (d1) {\(=\)};
    \vertex at (3.3,0) (d2) {\(\sum_{\nu} (R_{13})_{\mu \nu}\)};
    \vertex at (5,0) (na);
    \vertex at (4,-1) (nf1) {\(\sigma_1\)};
    \vertex at (4,1) (nf2) {\(\sigma_2\)};
    \vertex at (6,0) (nb);
    \vertex at (5.5,0.3) (nc) {\(\nu\)};
    \vertex at (7,1) (nf3) {\(\sigma_3\)};
    \vertex at (7,-1) (nf4) {\(\sigma_4\)};
    
    \diagram* {
      (f2) --  (a) --  (f3),
      (a)  -- (b),
      (f1) -- (b) -- (f4),
      (nf1) --  (na) --  (nf2),
      (na)  -- (nb),
      (nf3) -- (nb) -- (nf4)
    };
    \end{feynman}
    \end{tikzpicture}
\caption{\label{fig:III-2} Exchange matrices of the four-point correlator.  }
\end{figure}
The braiding matrix $R_{12}$ is diagonal because the two fields that we exchange fuse together. From the OPE, we deduce
\be \label{eq:III-10}   R_{12}^{(4)} =
 \begin{pmatrix} 
 e^{-i\pi \frac{k-1}{2(k+2)}} & 0 \\
 0 & e^{i\pi \frac{k+1}{2(k+2)}}
 \end{pmatrix}
 \;.\ee
The other two exchange matrices can be found using standard Hypergeometric and Gamma function identities. After some algebra, we obtain
\be \label{eq:III-11}   R_{13}^{(4)} =  \begin{pmatrix}
	\cos \theta_k & \sin \theta_k \\ \sin \theta_k & -\cos \theta_k
\end{pmatrix}
 \;,
\ee
where $\cos\theta_k = \frac{1}{2}\sec \frac{\pi}{k+2}$. The matrix $R_{23}$ can be deduced from
\be \label{eq:III-21} R_{23} = R_{13} R_{12} R^{-1}_{13} \ ,\ee 
We obtain
\be \label{eq:III-11a}   
 R_{23}^{(4)} =  \begin{pmatrix}
	e^{i \pi \frac{k-1}{2(k+2)}}\cos \theta_k & e^{-i \pi \frac{k+1}{2(k+2)}}\sin \theta_k \\ e^{-i \pi \frac{k+1}{2(k+2)}} \sin \theta_k & -e^{-i \pi \frac{3k+1}{2(k+2)}} \cos \theta_k
\end{pmatrix}
\;.
\ee
As an example, consider the $k=2$ case which corresponds to the critical Ising model. The diagonal braiding matrix becomes the  phase $S$ gate (up to a phase), while the fusion matrix reduces to the Hadamard gate, 
\be
\label{eq:III-13} 
R_{12}^{(4)} = e^{-i\frac{\pi}{8}} \begin{pmatrix}
	{1} & {0} \\ {0} & i
\end{pmatrix}
  \;\; , \;\;
R_{13}^{(4)} = \frac{1}{\sqrt{2}} \begin{pmatrix}
	{1} & {1} \\ {1} & -{1}
\end{pmatrix}
\ .
\ee
These matrices are not enough to achieve universal quantum computation \cite{bib:32} because we have no way to construct the phase $T$ gate using braidings.

Universal quantum computation can be achieved for $k=3$ which corresponds to the tri-critical Ising model. We obtain the matrices that appear in the Fibonacci anyon model 
\be
\label{eq:III-14} 
R_{12}^{(4)} =  \begin{pmatrix}
	e^{-i\frac{\pi}{5}} & 0 \\ 0 & e^{-i\frac{2\pi}{5}}
\end{pmatrix}
\ ,  \;
 R_{13}^{(4)} =  \begin{pmatrix}
	\frac{1}{\gamma} & \frac{1}{\sqrt{\gamma}} \\ \frac{1}{\sqrt{\gamma}} & -\frac{1}{\gamma}
\end{pmatrix}
\ ,
\ee
where $\gamma = \frac{\sqrt{5} +1}{2}$ is the golden ratio. The set \eqref{eq:III-14} is dense in $SU(2)$ \cite{bib:33}, leading to universal quantum computation. However, the minimal model $\mathcal{M} (5,4)$ cannot be used as a foundation for fault-tolerant quantum computation, because of the absence of a gapped state. In Section \ref{sec:V}, we will consider an alternative proposal using the coset $SU(2)^{\otimes 3}/ SU(2)_3$ that leads to universal quantum computation based on the braiding matrices \eqref{eq:III-14} as well as fault-tolerant quantum computation since $SU(2)^{\otimes 3}/ SU(2)_3$ possesses a gapped state.

\subsection{Five-point amplitudes}

\begin{figure}[t]
    \centering
    \begin{tikzpicture}
    \begin{feynman}
    \vertex at (0,0) (a);
    \vertex at (-1,-1) (f1) {\(\sigma_1\)};
    \vertex at (-1,1) (f2) {\(\sigma_2\)};
    \vertex at (1,0) (b);
    \vertex at (0.5,0.3) (d1) {\(\mathbb{I}\)};
    \vertex at (2,0) (c);
    \vertex at (1.5,0.3) (d2) {\(\varepsilon\)};
    \vertex at (3,1) (f3) {\(\sigma_3\)};
    \vertex at (3,-1) (f4) {\(\sigma_4\)};
    \vertex at (1,-1) (f5) {\(\varepsilon_5\)};
    \vertex at (-2,0)  (d3) {\(\mathcal{F}^{(5)}_1\)};
    \vertex at (-1.5,0)  (d4) {\(=\)};

    \diagram* {
      (f1) --  (a) --  (f2),
      (a) -- (b) -- (f5) -- (b) -- (c),
      (f3) -- (c) -- (f4)
    };
    \end{feynman}
    \end{tikzpicture}

    \begin{tikzpicture}
    \begin{feynman}
    \vertex at (0,0) (a);
    \vertex at (-1,-1) (f1) {\(\sigma_1\)};
    \vertex at (-1,1) (f2) {\(\sigma_2\)};
    \vertex at (1,0) (b);
    \vertex at (0.5,0.3) (d1) {\(\varepsilon\)};
    \vertex at (2,0) (c);
    \vertex at (1.5,0.3) (d2) {\(\mathbb{I}\)};
    \vertex at (3,1) (f3) {\(\sigma_3\)};
    \vertex at (3,-1) (f4) {\(\sigma_4\)};
    \vertex at (1,-1) (f5) {\(\varepsilon_5\)};
    \vertex at (-2,0)  (d3) {\(\mathcal{F}^{(5)}_2\)};
    \vertex at (-1.5,0)  (d4) {\(=\)};

    \diagram* {
      (f1) --  (a) --  (f2),
      (a) -- (b) -- (f5) -- (b) -- (c),
      (f3) -- (c) -- (f4)
    };
    \end{feynman}
    \end{tikzpicture}

    \begin{tikzpicture}
    \begin{feynman}
    \vertex at (0,0) (a);
    \vertex at (-1,-1) (f1) {\(\sigma_1\)};
    \vertex at (-1,1) (f2) {\(\sigma_2\)};
    \vertex at (1,0) (b);
    \vertex at (0.5,0.3) (d1) {\(\varepsilon\)};
    \vertex at (2,0) (c);
    \vertex at (1.5,0.3) (d2) {\(\varepsilon\)};
    \vertex at (3,1) (f3) {\(\sigma_3\)};
    \vertex at (3,-1) (f4) {\(\sigma_4\)};
    \vertex at (1,-1) (f5) {\(\varepsilon_5\)};
    \vertex at (-2,0)  (d3) {\(\mathcal{F}^{(5)}_3\)};
    \vertex at (-1.5,0)  (d4) {\(=\)};

    \diagram* {
      (f1) --  (a) --  (f2),
      (a) -- (b) -- (f5) -- (b) -- (c),
      (f3) -- (c) -- (f4)
    };
    \end{feynman}
    \end{tikzpicture}

    \caption{Conformal blocks of the five-point function.  In the case of the critical Ising model ($k=2$), the third conformal block vanishes since $\varepsilon \otimes \varepsilon = \mathbb{I}$.}
    \label{fig:III-3}
\end{figure}
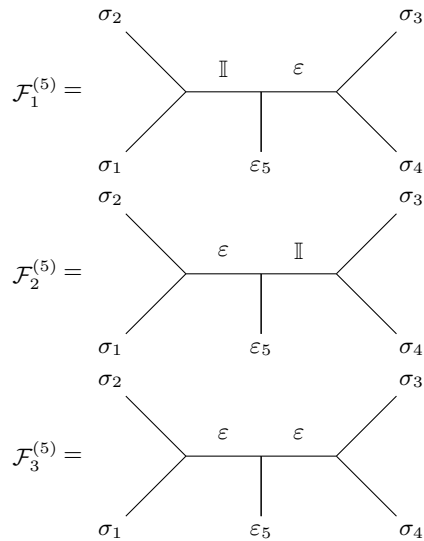

Next, we consider the five-point chiral amplitude of four $\sigma$ fields and one $\varepsilon$ field. This is not an amplitude of the type \eqref{eq:II-10} that we use for quantum computation. However, it is needed for the six-point chiral amplitude of $\sigma$ fields. 

The correlator needs a single negative screening charge in order to obey the charge neutrality condition,
%\be \label{eq:III-15}   G^{(5)} ( \bm{z}; w) = \int d^2 w | \mathcal{I}(z;w) |^2 \ , \ee
\be \label{eq:III-16} \mathcal{I}^{(5)}(\bm{\eta};w)= \left\langle \sigma_1 \sigma_2 \sigma_3 \sigma_4 \tilde{\varepsilon}_5 V_- (w)  \right\rangle \ . \ee
From the fusion rules \eqref{eq:II-4} we deduce that  there are two (three) conformal blocks for $k=2$ ($ k \geq 3$),
%Then we can write  the non-chiral correlator in terms of the conformal blocks
%\be \label{eq:III-17} G^{(5)} ( \bm{z}; w)  = |\mathcal{F}_1^{(5)}(z;w)|^2 + |\mathcal{F}_2^{(5)}(z;w)|^2 + |\mathcal{F}_3^{(5)}(z;w)|^2 \ .\ee
%These conformal blocks  are 
defined diagrammatically in Figure \ref{fig:III-3}, and in terms of contour integrals  by
\be \label{eq:III-18} \mathcal{F}_\mu^{(5)} (\bm{\eta}) =  \sqrt{N_\mu} \oint_{\mu} dw \mathcal{I}^{(5)}(\bm{\eta};w) \ .\ee 
The normalization constants $N_\mu$ are evaluated using a monodromy argument as before,
\bea
N_1 &=& \mathcal{N} \sin ^2\frac{2 \pi }{k+2} \ ,\nonumber \\
N_2 &=& \mathcal{N} \sin ^2 \frac{3 \pi }{k+2} \ ,\nonumber \\
N_3 &=&  8 \mathcal{N}\cos ^2 \frac{\pi }{k+2} \cos \frac{2 \pi }{k+2} \sin ^2 \frac{3 \pi }{k+2} \ ,
\eea
up to an overall multiplicative constant $\mathcal{N}$ which is not needed for our purposes.

\begin{figure}[ht!]
    \centering
    \begin{tikzpicture}
    \begin{feynman}
    \vertex at (0,0) (a);
    \vertex at (-0.7,-0.7) (f1) {\(\sigma_2\)};
    \vertex at (-0.7,0.7) (f2) {\(\sigma_1\)};
    \vertex at (1,0) (b);
    \vertex at (0.5,0) (c);
    \tiny
    \vertex at (0.25,0.25) (c1) {\( {\mu_1} \)};
    \vertex at (0.75,0.25) (c2) {\( \mu_2\)};
    \normalsize
    \vertex at (1.7,0.7) (f3) {\(\sigma_3\)};
    \vertex at (1.7,-0.7) (f4) {\(\sigma_4\)};
    \vertex at (0.5,-0.7) (f5) {\(\varepsilon_5\)};
    \vertex at (2.1,0) (d1) {\(=\)};
    \vertex at (3.2,0) (d2) {\( \sum_{\nu} (R_{12})_{\mu \nu} \)};
    \vertex at (5,0) (na);
    \vertex at (4.3,-0.7) (nf1) {\(\sigma_1\)};
    \vertex at (4.3,0.7) (nf2) {\(\sigma_2\)};
    \vertex at (6,0) (nb);
    \vertex at (5.5,0) (nc);
    \tiny
    \vertex at (5.25,0.25) (c3) {\( {\nu_1} \)};
    \vertex at (5.75,0.25) (c4) {\( \nu_2\)};
    \normalsize
    \vertex at (6.7,0.7) (nf3) {\(\sigma_3\)};
    \vertex at (6.7,-0.7) (nf4) {\(\sigma_4\)};
    \vertex at (5.5,-0.7) (nf5) {\(\varepsilon_5\)};

    \diagram* {
      (f1) --  (a) --  (f2),
      (a)  -- (c) -- (f5),
      (c) -- (b),
      (f3) -- (b) -- (f4),
      (nf1) --  (na) --  (nf2),
      (na)  -- (nc) -- (nf5),
      (nc) -- (nb),
      (nf3) -- (nb) -- (nf4)
    };
    \end{feynman}
    \end{tikzpicture}
    
\begin{tikzpicture}
    \begin{feynman}
    \vertex at (0,0) (a);
    \vertex at (-0.7,-0.7) (f1) {\(\sigma_1\)};
    \vertex at (-0.7,0.7) (f2) {\(\sigma_3\)};
    \vertex at (1,0) (b);
    \vertex at (0.5,0) (c);
    \tiny
    \vertex at (0.25,0.25) (c1) {\( {\mu_1} \)};
    \vertex at (0.75,0.25) (c2) {\( \mu_2\)};
    \normalsize
    \vertex at (1.7,0.7) (f3) {\(\sigma_2\)};
    \vertex at (1.7,-0.7) (f4) {\(\sigma_4\)};
    \vertex at (0.5,-0.7) (f5) {\(\varepsilon_5\)};
    \vertex at (2.1,0) (d1) {\(=\)};
    \vertex at (3.2,0) (d2) {\( \sum_{\nu} (R_{23})_{\mu \nu} \)};
    \vertex at (5,0) (na);
    \vertex at (4.3,-0.7) (nf1) {\(\sigma_1\)};
    \vertex at (4.3,0.7) (nf2) {\(\sigma_2\)};
    \vertex at (6,0) (nb);
    \vertex at (5.5,0) (nc);
    \tiny
    \vertex at (5.25,0.25) (c3) {\( {\nu_1} \)};
    \vertex at (5.75,0.25) (c4) {\( \nu_2\)};
    \normalsize
    \vertex at (6.7,0.7) (nf3) {\(\sigma_3\)};
    \vertex at (6.7,-0.7) (nf4) {\(\sigma_4\)};
    \vertex at (5.5,-0.7) (nf5) {\(\varepsilon_5\)};

    \diagram* {
      (f1) --  (a) --  (f2),
      (a)  -- (c) -- (f5),
      (c) -- (b),
      (f3) -- (b) -- (f4),
      (nf1) --  (na) --  (nf2),
      (na)  -- (nc) -- (nf5),
      (nc) -- (nb),
      (nf3) -- (nb) -- (nf4)
    };
    \end{feynman}
    \end{tikzpicture}    
    
    \begin{tikzpicture}
    \begin{feynman}
    \vertex at (0,0) (a);
    \vertex at (-0.7,-0.7) (f1) {\(\sigma_3\)};
    \vertex at (-0.7,0.7) (f2) {\(\sigma_2\)};
    \vertex at (1,0) (b);
    \vertex at (0.5,0) (c);
    \tiny
    \vertex at (0.25,0.25) (c1) {\( {\mu_1} \)};
    \vertex at (0.75,0.25) (c2) {\( \mu_2\)};
    \normalsize
    \vertex at (1.7,0.7) (f3) {\(\sigma_1\)};
    \vertex at (1.7,-0.7) (f4) {\(\sigma_4\)};
    \vertex at (0.5,-0.7) (f5) {\(\varepsilon_5\)};
    \vertex at (2.1,0) (d1) {\(=\)};
    \vertex at (3.2,0) (d2) {\( \sum_{\nu} (R_{13})_{\mu \nu} \)};
    \vertex at (5,0) (na);
    \vertex at (4.3,-0.7) (nf1) {\(\sigma_1\)};
    \vertex at (4.3,0.7) (nf2) {\(\sigma_2\)};
    \vertex at (6,0) (nb);
    \vertex at (5.5,0) (nc);
    \tiny
    \vertex at (5.25,0.25) (c3) {\( {\nu_1} \)};
    \vertex at (5.75,0.25) (c4) {\( \nu_2\)};
    \normalsize
    \vertex at (6.7,0.7) (nf3) {\(\sigma_3\)};
    \vertex at (6.7,-0.7) (nf4) {\(\sigma_4\)};
    \vertex at (5.5,-0.7) (nf5) {\(\varepsilon_5\)};

    \diagram* {
      (f1) --  (a) --  (f2),
      (a)  -- (c) -- (f5),
      (c) -- (b),
      (f3) -- (b) -- (f4),
      (nf1) --  (na) --  (nf2),
      (na)  -- (nc) -- (nf5),
      (nc) -- (nb),
      (nf3) -- (nb) -- (nf4)
    };
    \end{feynman}
    \end{tikzpicture}
    
    \caption{Exchange matrices of the five-point correlator. In our notation each channel $\mu$ is represented by a pair $(\mu_1, \mu_2)$. For the first channel we have $(\mathbb{I},\varepsilon)$, for the second  $(\varepsilon, \mathbb{I})$, and the third   $(\varepsilon,\varepsilon)$.}
    \label{fig:III-4}
\end{figure}
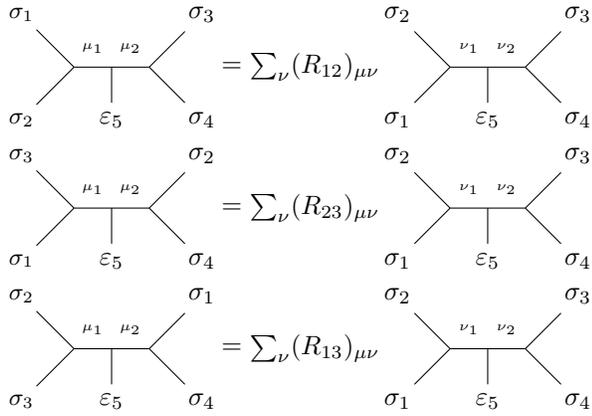

The braiding and fusion matrices for the five-point amplitude are depicted in Figure \ref{fig:III-4}.
The braiding matrix $R_{12}$ is easily obtained from the OPE  $ \sigma (\eta) \sigma(0) \sim \eta^{-\frac{k-1}{2(k+2)}} \mathbb{I} + \eta^{ \frac{k+1}{2(k+2)}} \varepsilon (0) $,
\be \label{eq:III-19}  R_{12}^{(5)} = e^{-i\pi \frac{k-1}{2(k+2)}}
 \begin{pmatrix} 
 1 & 0 & 0 \\
 0 & e^{i\pi \frac{k}{k+2}} & 0 \\
 0 & 0& e^{i\pi \frac{k}{k+2}} 
 \end{pmatrix}
 \;.\ee
The fusion matrix $R_{13}$ is found by converting five-point functions into four-point functions (see Appendix \ref{app:A} for details).
We obtain
\be
\label{eq:III-20}
R_{13}^{(5)} = \begin{pmatrix} 
 c_k &  c_k& -\sqrt{d_k} \\
 c_k & \frac{d_k \omega_k -c_k^3}{s_k^2} & \frac{ (\omega_k+c_k) \sqrt{d_k}}{s_k t_k}  \\
-\sqrt{d_k}  & \frac{ (\omega_k+c_k) \sqrt{d_k}}{s_k t_k} &\frac{\omega_k c_k -d_k}{s_k t_k} 
\end{pmatrix} \ ,
\ee
where $c_k = \cos\theta_k$, $d_k= -\cos 2\theta_k$, $s_k = \sin\theta_k$, $t_k = \tan\theta_k$, and $\omega_k = e^{i\pi \frac{3(k+1)}{2(k+2)}}$.
The braiding matrix $R_{23}$ is deduced from \eqref{eq:III-21}.

As an example, for the tri-critical Ising model ($k=3$), we obtain
\be \label{eq:III-19f}  R_{12}^{(5)} = 
 \begin{pmatrix} 
 e^{-i \frac{2\pi}{5}} & 0 & 0 \\
 0 & e^{i \frac{\pi}{5}} & 0 \\
 0 & 0& e^{i \frac{\pi}{5}} 
 \end{pmatrix}
 \;,\ee
\be \label{eq:III-22} R_{13}^{(5)} = 
 \begin{pmatrix} 
 \gamma^{-1} & \gamma^{-1}   & -\gamma^{-3/2}   \\
 \gamma^{-1} & -\frac{ \left({1+e^{i\frac{\pi}{5}}}\right)}{\gamma^2}& -\gamma^{-5/2} e^{i\frac{2\pi}{5}}\\ 
 -\gamma^{-3/2} & -\gamma^{-5/2} e^{i\frac{2\pi}{5}}&  - \frac{\left({\gamma^{-1} +\gamma e^{i\frac{\pi}{5}}}\right) }{\gamma^2}
 \end{pmatrix}\ ,\ee
 and
 \be \label{eq:III-23} R_{23}^{(5)} = 
 \begin{pmatrix} 
 e^{i\frac{\pi}{5}} \gamma^{-1} & e^{-i\frac{2\pi}{5}}\gamma^{-1}   & -e^{-i\frac{2\pi}{5}} \gamma^{-3/2}   \\
e^{-i\frac{2\pi}{5}}\gamma^{-1}  &  e^{i\frac{\pi}{5}} \gamma^{-1} & -e^{-i\frac{2\pi}{5}} \gamma^{-3/2} \\ 
 -e^{-i\frac{2\pi}{5}} \gamma^{-3/2}  & -e^{-i\frac{2\pi}{5}} \gamma^{-3/2} &  e^{i\frac{\pi}{5}}-\gamma^{-2}
 \end{pmatrix}\ .\ee
 These expressions are needed for six-point amplitudes to be discussed next.
 
\subsection{Six-point amplitudes}
Next, we consider the amplitude involving six $\sigma$ fields.
From the fusion rules \eqref{eq:II-4}, we know that there are 4 conformal blocks for $k=2$ and 5 conformal blocks for $k \geq 3$, as shown in Figure \ref{fig:III-5}.

\begin{figure}[ht!]
\centering
\begin{tikzpicture}
\begin{feynman}
    \vertex at (0,0) (a);
    \vertex at (-0.7,-0.7) (f1) {\(\sigma_1\)};
    \vertex at (-0.7,0.7) (f2) {\(\sigma_2\)};
    \vertex at (1,0) (b);
    \vertex at (1.7,0.7) (c);
    \vertex at (1.7,-0.7) (d);
    \vertex at (1.6,1.7) (f3) {\(\sigma_3\)};
    \vertex at (2.7,0.9) (f4) {\(\sigma_4\)};
    \vertex at (2.7,-0.9) (f5) {\(\sigma_5\)};
    \vertex at (1.6,-1.7) (f6) {\(\sigma_6\)};
    \vertex at (0.5,0.3) (k1) {\(\mathbb{I}\)};
    \vertex at (1.6,0.2) (k2) {\(\mathbb{I}\)};
    \vertex at (1.2,-0.6) (k3) {\(\mathbb{I}\)};
    \vertex at (-1.8,0)  (d1) {\(\mathcal{F}^{(6)}_1\)};
    \vertex at (-1.2,0)  (d2) {\(=\)};
    \diagram* {
      (f1) --  (a) --  (f2),
      (a)  -- (b),
      (f3) --  (c) --  (f4),
      (c) -- (b),
      (f5) -- (d) -- (f6),
      (d) -- (b)
    };
\end{feynman}
\end{tikzpicture}

\begin{tikzpicture}
\begin{feynman}
    \vertex at (0,0) (a);
    \vertex at (-0.7,-0.7) (f1) {\(\sigma_1\)};
    \vertex at (-0.7,0.7) (f2) {\(\sigma_2\)};
    \vertex at (1,0) (b);
    \vertex at (1.7,0.7) (c);
    \vertex at (1.7,-0.7) (d);
    \vertex at (1.6,1.7) (f3) {\(\sigma_3\)};
    \vertex at (2.7,0.9) (f4) {\(\sigma_4\)};
    \vertex at (2.7,-0.9) (f5) {\(\sigma_5\)};
    \vertex at (1.6,-1.7) (f6) {\(\sigma_6\)};
    \vertex at (0.5,0.3) (k1) {\(\varepsilon\)};
    \vertex at (1.6,0.2) (k2) {\(\varepsilon\)};
    \vertex at (1.2,-0.6) (k3) {\(\mathbb{I}\)};
    \vertex at (-1.8,0)  (d1) {\(\mathcal{F}^{(6)}_2\)};
    \vertex at (-1.2,0)  (d2) {\(=\)};
    \diagram* {
      (f1) --  (a) --  (f2),
      (a)  -- (b),
      (f3) --  (c) --  (f4),
      (c) -- (b),
      (f5) -- (d) -- (f6),
      (d) -- (b)
    };
\end{feynman}
\end{tikzpicture}

\begin{tikzpicture}
\begin{feynman}
    \vertex at (0,0) (a);
    \vertex at (-0.7,-0.7) (f1) {\(\sigma_1\)};
    \vertex at (-0.7,0.7) (f2) {\(\sigma_2\)};
    \vertex at (1,0) (b);
    \vertex at (1.7,0.7) (c);
    \vertex at (1.7,-0.7) (d);
    \vertex at (1.6,1.7) (f3) {\(\sigma_3\)};
    \vertex at (2.7,0.9) (f4) {\(\sigma_4\)};
    \vertex at (2.7,-0.9) (f5) {\(\sigma_5\)};
    \vertex at (1.6,-1.7) (f6) {\(\sigma_6\)};
    \vertex at (0.5,0.3) (k1) {\(\mathbb{I}\)};
    \vertex at (1.6,0.2) (k2) {\(\varepsilon\)};
    \vertex at (1.2,-0.6) (k3) {\(\varepsilon\)};
    \vertex at (-1.8,0)  (d1) {\(\mathcal{F}^{(6)}_3\)};
    \vertex at (-1.2,0)  (d2) {\(=\)};
    \diagram* {
      (f1) --  (a) --  (f2),
      (a)  -- (b),
      (f3) --  (c) --  (f4),
      (c) -- (b),
      (f5) -- (d) -- (f6),
      (d) -- (b)
    };
\end{feynman}
\end{tikzpicture}

\begin{tikzpicture}
\begin{feynman}
    \vertex at (0,0) (a);
    \vertex at (-0.7,-0.7) (f1) {\(\sigma_1\)};
    \vertex at (-0.7,0.7) (f2) {\(\sigma_2\)};
    \vertex at (1,0) (b);
    \vertex at (1.7,0.7) (c);
    \vertex at (1.7,-0.7) (d);
    \vertex at (1.6,1.7) (f3) {\(\sigma_3\)};
    \vertex at (2.7,0.9) (f4) {\(\sigma_4\)};
    \vertex at (2.7,-0.9) (f5) {\(\sigma_5\)};
    \vertex at (1.6,-1.7) (f6) {\(\sigma_6\)};
    \vertex at (0.5,0.3) (k1) {\(\varepsilon\)};
    \vertex at (1.6,0.2) (k2) {\(\mathbb{I}\)};
    \vertex at (1.2,-0.6) (k3) {\(\varepsilon\)};
    \vertex at (-1.8,0)  (d1) {\(\mathcal{F}^{(6)}_4\)};
    \vertex at (-1.2,0)  (d2) {\(=\)};
    \diagram* {
      (f1) --  (a) --  (f2),
      (a)  -- (b),
      (f3) --  (c) --  (f4),
      (c) -- (b),
      (f5) -- (d) -- (f6),
      (d) -- (b)
    };
\end{feynman}
\end{tikzpicture}

\begin{tikzpicture}
\begin{feynman}
    \vertex at (0,0) (a);
    \vertex at (-0.7,-0.7) (f1) {\(\sigma_1\)};
    \vertex at (-0.7,0.7) (f2) {\(\sigma_2\)};
    \vertex at (1,0) (b);
    \vertex at (1.7,0.7) (c);
    \vertex at (1.7,-0.7) (d);
    \vertex at (1.6,1.7) (f3) {\(\sigma_3\)};
    \vertex at (2.7,0.9) (f4) {\(\sigma_4\)};
    \vertex at (2.7,-0.9) (f5) {\(\sigma_5\)};
    \vertex at (1.6,-1.7) (f6) {\(\sigma_6\)};
    \vertex at (0.5,0.3) (k1) {\(\varepsilon\)};
    \vertex at (1.6,0.2) (k2) {\(\varepsilon\)};
    \vertex at (1.2,-0.6) (k3) {\(\varepsilon\)};
    \vertex at (-1.8,0)  (d1) {\(\mathcal{F}^{(6)}_5\)};
    \vertex at (-1.2,0)  (d2) {\(=\)};
    \diagram* {
      (f1) --  (a) --  (f2),
      (a)  -- (b),
      (f3) --  (c) --  (f4),
      (c) -- (b),
      (f5) -- (d) -- (f6),
      (d) -- (b)
    };
\end{feynman}
\end{tikzpicture}
\caption{Conformal blocks of the six-point function. For $k=2$ the last conformal block vanishes.}
\label{fig:III-5}
\end{figure}
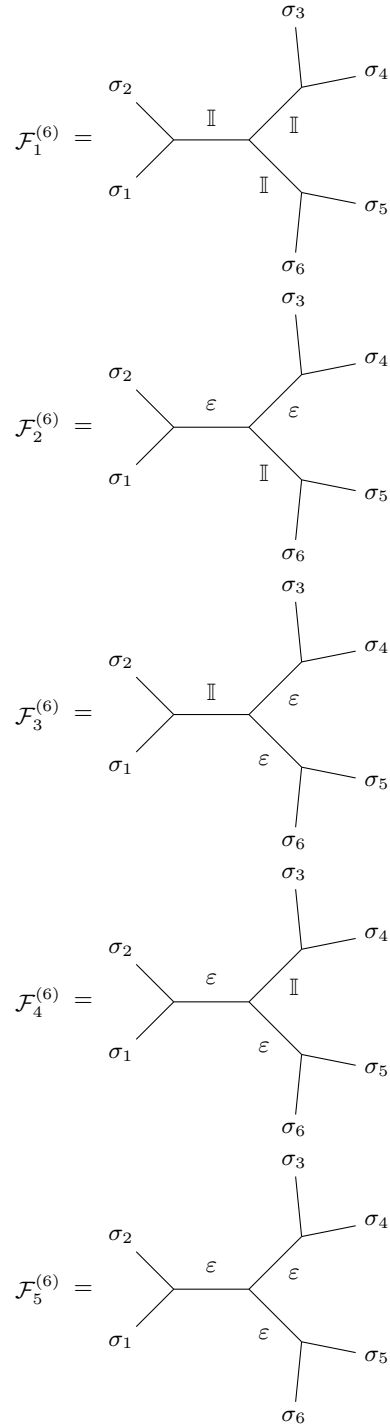
Using the OPE $ \sigma (\eta_5) \sigma(\eta_6) \sim \eta_{56}^{-\frac{k-1}{2(k+2)}} \mathbb{I} + \eta_{56}^{ \frac{k+1}{2(k+2)}} \varepsilon (\eta_5) $ to expand near $\eta_6=\eta_5$, we notice two different subspaces, one for the $\frac{1-k}{2(k+2)}$ and one for the $\frac{k+1}{2(k+2)}$ powers of $\eta_{56}$. The first one contains $\mathcal{F}^{(6)}_1$ and $\mathcal{F}^{(6)}_2$ is similar to the four-point amplitude, whereas the second one contains $\mathcal{F}^{(6)}_3$, $\mathcal{F}^{(6)}_4$ and $\mathcal{F}^{(6)}_5$ and is similar to the five-point amplitude.

The exchange matrices corresponding to the exchanges $ \eta_1 \leftrightarrow  \eta_2$, $ \eta_1 \leftrightarrow  \eta_3$ and $ \eta_2 \leftrightarrow  \eta_3$ can be found using the four-point and five-point matrices, 
\be R^{(6)} = \left(\begin{array}{@{}c|c@{}}
  R^{(4)}
  & \bm{0} \\
\hline
  \bm{0} &
  R^{(5)}
\end{array}\right) \ , \ \ R \in \{ R_{12}, R_{13}, R_{23} \}
\ee
For the critical Ising model $(k=2)$ the $4\times4$ exchange matrices have been studied in \cite{bib:34,bib:35,bib:36}. Confirming these results, we observe that the last conformal block decouples and the remaining $2\times2$ blocks correspond to a system of two qubits,
\be R^{(6)} = \left(\begin{array}{@{}c|c|c@{}}
  R^{(4)}
  & \bm{0} & 0 \\
\hline
  \bm{0} &
  R^{(4)} & 0 \\
  \hline
  \bm{0} & \bm{0} & r
\end{array}\right) \ , \ \ R \in \{ R_{12}, R_{13}, R_{23} \}
\ee
where $r$ is an irrelevant phase. These matrices are gates acting on two qubits. However, they do not lead to universal quantum computation.

For $k\ge 3$, the exchange matrices form a sufficient set of gates for universal quantum computation in five dimensions.
 
 Although we focused the discussion on exchange matrices $R_{ij}$, $i,j=1,2,3$, the above method can be straightforwardly extended to include the point $\eta_4$. To obtain exchange matrices involving the points $\eta_5$ or $\eta_6$, we need to consider different limits that reduce the six-point amplitude to different four- and five-point amplitudes. For example, to calculate the exchange matrix $R_{15}$ we can expand near $\eta_4=\eta_3$ using the OPE $ \sigma (\eta_4) \sigma(\eta_3) \sim \eta_{34}^{-\frac{k-1}{2(k+2)}} \mathbb{I} + \eta_{34}^{ \frac{k+1}{2(k+2)}} \varepsilon (\eta_3) $. We obtain two distinct subspaces, one corresponding to the four-point results obtained earlier, but with conformal blocks $\mathcal{F}_1^{(6)}$ and $\mathcal{F}_4^{(6)}$, and the other corresponding to a five-point amplitude  with conformal blocks $\mathcal{F}_2^{(6)}$, $\mathcal{F}_3^{(6)}$, and $\mathcal{F}_5^{(6)}$. All other exchange matrices are constructed similarly. 
 
 \subsection{Higher-point amplitudes}
 
 The dimensionality of Hilbert space (number of conformal blocks) depends on both $N$ and $k$. For the four-point amplitude $(N=2)$ we have two conformal blocks for all $k$, due to the fusion rule $\sigma\times \sigma \sim \varepsilon$ (Eq.\ \eqref{eq:II-4} with $s=s'=2$). For the six-point amplitude $(N=3)$, we have four conformal blocks for $k=2$ and five conformal blocks for all other cases $(k\geq 3)$. Using the fusion rules \eqref{eq:II-4}, we can find the number of conformal blocks for higher-point amplitudes. In particular, for $N=4$ we have 8 conformal blocks for $k=2$, 13 for $k=3$,  and 14 for all other cases $(k\geq 4)$. For $k=2$ (critical Ising model), the dimensionality of Hilbert space is $2^{N-1}$, whereas for $k=3$ (tricritical Ising model), it follows the Fibonacci sequence. General expressions for other $k\ge 4$ can also be found using the fusion rules.
 
 Although higher-point amplitudes cannot be explicitly calculated, we can still obtain the exchange matrices by
 following the procedure discussed above for the six-point amplitude. For example, to  find the matrices  $R_{12}$, $R_{13}$ and $R_{23}$ for the eight-point amplitude, we will work in the limit $\eta_8\to \eta_7$ and $\eta_6\to \eta_5$. We obtain the exchange matrix for the eight-point function as a block diagonal matrix, with each block corresponding to a four- or five- function. This procedure can be generalized to arbitrary $N$.

\section{Coset Amplitudes} \label{sec:coamp}

Correlators in the coset CFT $SU(2)^{\otimes k}/SU(2)_k$ can be factorized \cite{bib:37,bib:37a} in products of correlators of $SU(2)_1$ and $\overline{SU(2)_k}$ WZW models. The primary fields in these WZW models are $\chi_m^{[i]j}$ ($i=1,\dots, k$) in $SU(2)_1$ and $\bar{\tau}_m^j$ in $\overline{SU(2)_k}$. We are interested in the case $j=\frac{1}{2}$. The conformal weights of these primary fields are
\be\label{eq:35} h_\chi = \frac{1}{4}\ , \ \ h_{\bar{\tau}} = -\frac{3}{4(k+2)}\ee
To simplify the notation, we will drop the index $j$ and use $m=\pm$ to denote $m= \pm\frac{1}{2}$. Then correlators are of the general form
\bea \label{eq:V-6}
X^{[i]}_{m_1 \dots m_N}&=&
\braket{\chi^{[i]}_{m_1}(\eta_1) \cdots \chi^{[i]}_{m_N}(\eta_N)} \ , \ \ i=1,\dots, k \nonumber \\
Y_{\mu ;m_1 \dots m_N}&=& \braket{\bar{\tau}_{m_1}(\eta_1)\cdots \bar{\tau}_{m_N}(\eta_N)} \ .
\eea
where $\mu$ labels the corresponding conformal block.

The primary field $\sigma = \Phi_{(1,2)}$ can be constructed using
\be \label{eq:V-3a}
\sigma^{[i]} =  \chi^{[i]}_+ \bar{\tau}_+ + \chi^{[i]}_- \bar{\tau}_- \ .
\ee
Evidently, this does not lead to a unique definition since we can consider any of the $SU(2)_1$ factors in the coset to construct $\sigma$. We will identify $\sigma \equiv \sigma^{[k]}$. Using \eqref{eq:35}, we obtain its conformal weight $h_\sigma = h_\chi + h_{\bar\tau} = \frac{k-1}{4(k+2)}$, in agreement with the minimal model result in Table \ref{table:1}.

%{\color{red} In general we can construct all $\Phi_{(1,s)}$ with $1\leq s \leq k+1$ primary fields in the following way
%\be
%\Phi_{(1,s)}= \left\{
%\begin{array}{cc}
%    \bar{\tau}_+^\frac{s-1}{2}+ \bar{\tau}_-^\frac{s-1}{2} & s: \text{odd} \\
%                \\
%    \chi^{[k]}_+\bar{\tau}_+^\frac{s-1}{2}+ \chi^{[k]}_-\bar{\tau}_-^\frac{s-1}{2} & s: \text{even}
%\end{array}\right
%\ee
%and their conformal dimension is given by 
%\be
%h_{(1,s)}= \left\{
%\begin{array}{cc}
%    \Delta -\frac{(s-1)(s+1)}{4(k+2)} & s: \text{odd} \\
%                \\
%    \Delta + \frac{1}{4}-\frac{(s-1)(s+1)}{4(k+2)}  & s: \text{even}
%\end{array}\right
%\ee
%where $\Delta= (\left{\frac{s-1}{2}}\right)^2$ for odd $s$ and $\Delta= (\left{\frac{s-1}{2} -\frac{1}{4}}\right)^2-\frac{1}{4}$ for even $s$.
%}

Agreement with the minimal model $\mathcal{M} (k+2,k+1)$ is expected, because the latter can be constructed from the coset $SU(2)_{k-1}\times SU(2)_1/SU(2)_k$. The field $\sigma$ in the minimal model is also given by \eqref{eq:V-3a} with $\chi_\pm$ in the (single) $SU(2)_1$ factor in the coset $SU(2)_{k-1}\times SU(2)_1/SU(2)_k$. Therefore, correlators of the $\sigma$ field agree in the two CFTs (minimal model and coset $SU(2)^{\otimes k}/SU(2)_k$).

The two chiral conformal blocks for the four-point amplitude $\braket{\sigma_1 \sigma_2\sigma_3  \sigma_4 }$ are found from
\bea \label{eq:V-5x}
\mathcal{F}^{(4)}_\mu(\bm{\eta})&=& \sqrt{N_{\mu}} \sum_{m_1 \dots m_4} X_{m_1 \dots m_4} Y_{\mu;m_1 \dots m_4} \ ,
\eea
where we dropped the index $[k]$ that does not affect the calculation.
%where $N_{\mu}$ is the normalization constant derived by the monodromy around $0$ and $1$ as discussed in previous sections and
To evaluate these correlators, we will use the free field representation of $SU(2)_k$ \cite{bib:38,bib:39,bib:40}. The primary fields are defined in terms of a massless free boson $\varphi$  and a $( \beta, \gamma)$ bosonic ghost system. Correlators are evaluated using the Coulomb gas formalism. Charge neutrality is enforced  using screening charges and conjugate fields, as needed. It leads to the constraint that the total magnetic number vanishes ($\sum m_i = 0$). %The summation over $m_i$ with $1\leq i \leq 4$ takes into account only those states with zero net $m$. To simplify the equations we denote the $m=\pm \frac{1}{2}$ by $m=\pm$. Therefore we can write
The sum in \eqref{eq:V-5x} reduces to
\bea \label{eq:V-7}
\mathcal{F}^{(4)}_\mu(\bm{\eta})=&& 2 \sqrt{N_{\mu}} [  X_{+--+} Y_{\mu;+--+} 
+ X_{--++} Y_{\mu;--++} \nonumber\\ && + X_{-+-+} Y_{\mu;-+-+} ] \ .
\eea
As before, a global conformal transformation fixes three points, $\eta_1 \rightarrow 0$, $\eta_2 \rightarrow x$, $\eta_3 \rightarrow 1$ and $\eta_4 \rightarrow \infty$.  See Appendix \ref{app:B} for a detailed calculation  of $SU(2)_k$ correlators. For $SU(2)_1$ we obtain the functions
\be \label{eq:V-8}
X_{+--+} =  \mathcal{C} \sqrt{\frac{1-x}{x}} \ , \
X_{--++} = \mathcal{C} \sqrt{\frac{x}{1-x}} \ , \ee 
where $\mathcal{C} = \frac{4 \pi^2  }{\Gamma \left(\frac{1}{3}\right)^3}$, and
\be\label{eq:39} X_{-+-+} = - X_{+--+} - X_{--++} \ ,
\ee
by setting $q=1$ in the corresponding expressions \eqref{eq:B5} - \eqref{eq:B7} \cite{bib:25}.
%Notice that the sum of these three functions vanishes.

Similarly, we obtain the functions for the two $\overline{SU(2)_k}$ conformal blocks by setting $q=-k-4$ in the corresponding expressions \eqref{eq:B5} - \eqref{eq:B10} \cite{bib:25},
\begin{widetext}
\bea \label{eq:V-10}
Y_{1;+--+} &=& -\frac{\Gamma \left(\frac{1}{k+2}\right) \Gamma \left(\frac{k+3}{k+2}\right)}{\Gamma \left(\frac{k+4}{k+2}\right)}  (1-x)^{-\frac{1}{2 k+4}} x^{\frac{3}{2 k+4}}   \;_2F_1\left(-\frac{1}{k+2},\frac{1}{k+2};\frac{k+4}{k+2};x\right) \nonumber \ , \\
Y_{1;--++} &=& \frac{\Gamma \left(\frac{1}{k+2}\right) \Gamma \left(\frac{k+3}{k+2}\right) }{(k+4) \Gamma \left(\frac{k+4}{k+2}\right)} (1-x)^{-\frac{1}{2 k+4}} x^{\frac{2 k+7}{2 k+4}}   \;_2F_1\left(\frac{k+1}{k+2},\frac{k+3}{k+2};\frac{2 (k+3)}{k+2};x\right) \nonumber \ ,\\
Y_{2;+--+} &=& \frac{ \Gamma \left(\frac{1}{k+2}\right)  \Gamma \left(-\frac{3}{k+2}\right) }{2 \Gamma \left(-\frac{2}{k+2}\right)}  ((1-x) x)^{-\frac{1}{2 k+4}}  \;_2F_1\left(-\frac{3}{k+2},-\frac{1}{k+2};\frac{k}{k+2};x\right)  \ , \nonumber \\
Y_{2;--++}  &=&-\frac{ \Gamma \left(-\frac{3}{k+2}\right) \Gamma \left(\frac{1}{k+2}\right) }{\Gamma \left(-\frac{2}{k+2}\right)} ((1-x) x)^{-\frac{1}{2 k+4}}  \;_2F_1\left(-\frac{3}{k+2},-\frac{1}{k+2};-\frac{2}{k+2};x\right) \ , \nonumber\\
Y_{\mu;-+-+} &=& - Y_{\mu;+--+}  - Y_{\mu;--++} \ , \ \ \mu=1,2 \ .
\eea
\end{widetext}
After some algebra involving Hypergeometric function identities, we arrive at compact explicit expressions for the conformal blocks \eqref{eq:V-5x},
\bea \label{eq:V-12}
\mathcal{F}^{(4)}_1
&=& - \sqrt{N_1} \frac{32 \pi ^2  \Gamma^2 \left(\frac{1}{k+2}\right) }{\Gamma^3 \left(\frac{1}{3}\right) \Gamma \left(\frac{2}{k+2}\right)}  x^{\frac{1-k}{2 k+4}} (1-x)^{\frac{k+1}{2 k+4}}  \nonumber \\
&& \times \, _2F_1\left(\frac{1}{k+2},\frac{k+1}{k+2};\frac{2}{k+2};x\right) \ , \nonumber \\
\mathcal{F}^{(4)}_2
&=& \sqrt{N_2} \frac{12 \pi ^2 (1-k) \Gamma \left(-\frac{3}{k+2}\right) \Gamma \left(\frac{1}{k+2}\right) }{k \Gamma^3 \left(\frac{1}{3}\right) \Gamma \left(-\frac{2}{k+2}\right)} x^{\frac{k+1}{2 k+4}} \nonumber \\
&& \times  \, (1-x)^{\frac{k+1}{2 k+4}}  _2F_1\left(\frac{k+1}{k+2},\frac{2 k+1}{k+2};\frac{2 (k+1)}{k+2};x\right)\ , \nonumber\\
\eea
in agreement with our earlier result \eqref{eq:III-8}. It follows that the exchange matrices one obtains from the coset construction coincide with their counterparts in the corresponding minimal model. It is instructive to confirm this using Eq.\ \eqref{eq:V-7} in order to obtain exchange matrices for general correlators in the coset CFT $SU(2)_1^{\otimes k}/SU(2)_k$.

As an example, under the transformation $x\to 1-x$, it is easy to see that $X_{\bm{m}} \leftrightarrow X_{\bm{m}'}$ and $X_{\bm{m}''} \to X_{\bm{m}''}$, where $\bm{m} = +--+$, $\bm{m}' = --++$, and $\bm{m}'' = -+-+$. Also, using Hypergeometric identities, we obtain
\be\label{eq:42} \bm{Y}_{\bm{m}} \to R_{13}^{(4)} \bm{Y}_{\bm{m}'} \ , \ \bm{Y}_{\bm{m}'} \to R_{13}^{(4)} \bm{Y}_{\bm{m}} \ , \ \bm{Y}_{\bm{m}''} \to R_{13}^{(4)} \bm{Y}_{\bm{m}''} \ee
where $\bm{Y}_{\bm{m}} = \left( \begin{array}{c}
   \sqrt{N_1} Y_{1;\bm{m}}  \\
    \sqrt{N_2} Y_{2;\bm{m}}
\end{array}
\right)$, and $R_{13}^{(4)}$ is defined in \eqref{eq:III-11}. It follows from \eqref{eq:V-7} that the conformal blocks transform under
\be\label{eq:43} \bm{\mathcal{F}}^{(4)} \to R_{13}^{(4)} \bm{\mathcal{F}}^{(4)} \ , \ \ \bm{\mathcal{F}}^{(4)} = \left( \begin{array}{c}
   \mathcal{F}_1^{(4)}  \\
    \mathcal{F}_2^{(4)}
\end{array}
\right) \ee 
as expected.

\section{Braiding conformal blocks via anyon models} \label{sec:IV}

The monodromy representations from braiding conformal blocks can also be computed using the corresponding anyon models of chiral minimal models, which are the representation categories of the chiral algebras of minimal models.  In this Section, we calculate some of the same representations in earlier sections using the graphical calculus of anyon models.

\subsection{Anyon models of minimal models}

As discussed above, chiral minimal models $\mathcal{M}(k+2,k+1)$ ($k\geq 2$) can be constructed as the coset $\frac{SU(2)_{k-1}\times SU(2)_1}{SU(2)_{k}}$, where $SU(2)_q$ is the $SU(2)$ WZW model at level $q$.  We will also use $SU(2)_q$ to denote the corresponding anyon models of chiral WZW $SU(2)_q$ CFTs.

The anyons are sometimes labeled by integers $0, 1, \ldots, q$, which are twice of the spin, and fusion rules in these labels are
\be
	j_1 \otimes j_2 = \sum_{j \stackrel{2}{=} |j_1 - j_2|}^{\min(j_1+j_2, 2q-j_1-j_2)} j,
\ee
where $\stackrel{2}{=}$ denotes incrementing the summation variable by $2$, and twist
\be
	\theta_j = e^{\pi i \frac{j\pn{j+2}}{2(q+2)}}.
\ee
As discussed in Section \ref{sec:II}, the minimal model $\M(k+2,k+1)$ has primary fields $\Phi_{(r,s)}$ with fusion rules given by Eq.\ \eqref{eq:fusion}.
\begin{lemma}\label{lemma:k-r}
	In $SU(2)_q$, $q \otimes s = r$ if and only if $s = q - r$.  Moreover, when $s \neq q - r$, the product $q \otimes s$ contains no $r$ term.
\end{lemma}
\begin{proof}
	\begin{enumerate}
		\item[($\Leftarrow$)]
			Observe that
			\[
				q \otimes (q-r) = \sum_{j \stackrel{2}{=} r}^{\min(2k-r, r)} j = r \ ,
			\]
since $\min(2q-r,r) = r$, for all $r = 0, \ldots, q$.
		\item[($\Rightarrow$)]
			We have
			\[
				q \otimes s = \sum_{j \stackrel{2}{=} q - s}^{\min(q+s, q-s)} j.
			\]
			If $s < q - r$, then $q - s > r$, and there is no $r$ term in the product $q \otimes s$.  If $s > q - r$, then $q - s < r$, so $\min(q+s,q-s) < r$, and there is no $r$ term in the product $q \otimes s$.
	\end{enumerate}
\end{proof}

\begin{prop}\label{prop:fusionRules}
	If $\B$ is the anyon model $\B= SU(2)_{k-1} \times SU(2)_1 \times \overline{SU(2)_{k}}$ and $\A = 000 + (k-1)1k$, then $\A$ has a condensable algebra structure. The condensed category is $\B_{\A} = \B_0 \oplus \B_1$, where the deconfined part $\B_0$ has the same fusion rules as the Minimal Model $\M(k+2,k+1)$.  The twists of the anyons of $\B_0$ agree with those of the corresponding ones in the minimal model $\M(k+2,k+1)$.
\end{prop}
\begin{proof}

			For each $r = 0, \ldots, k-1$ and $t = 0, \ldots, k$, we may uniquely choose $s = 0$ or $s = 1$ so that $r+s+t$ is even.  Then, we may identify $rst \sim \Phi_{(r,t)}$ in $\M(k+2,k+1)$.  We have
			\[
				rst \otimes mnp = \sum_{j \stackrel{2}{=} |r-m|}^{\min(r+m, 2k-2-r-m)} \sum_{l \stackrel{2}{=} |t-p|}^{\min(t+p, 2k-t-p)} jsl,
			\]
			where $s$ is chosen to make $j+s+l$ even, and
			\[
				\Phi_{(r,t)} \otimes \Phi_{(m,p)} = \sum_{j \stackrel{2}{=} |m-r|}^{\min(m+r, 2k-2-m-r)} \sum_{l \stackrel{2}{=} |p-t|}^{\min(p+t, 2k-p-t)} \Phi_{(j,l)}.
			\]
%	\end{enumerate}
\end{proof}
The conformal weights are given by Eq.\ \eqref{eq:II-3}. 
%(page 208 of \cite{bib:28}),
%\[
%	h_{r,s}(m) = \frac{\bk{(m+1)r - ms}^2 - 1}{4m(m+1)}.
%\]
%Indexing from zero gives
%\[
%	h_{r,s}(m) = \frac{\bk{(m+1)(r+1) - m(s+1)}^2 - 1}{4m(m+1)}.
%\]
%In $\M(k+3,k+2)$, we take $m = k+2$.
\footnote{Page 240 of \cite{bib:28} gives another formula for $h$ that seems to suggest $k+1$ and $k+2$ should be switched.  The choice we made here makes the twists agree.  Possibly relevant is Eq.\ (7.36) on page 209.}

\subsection{Braiding universality of minimal models}

The anyon model of the tricritical Ising model $\mathcal{M}(5,4)$ is the direct product of the Ising anyon model with the complex conjugate of the Fibonacci anyon model.  The anyon types of the Ising theory are usually denoted as $1,\sigma,\psi$, where $\sigma$ is the non-abelian anyon, while the anyon types of the Fibonacci are $1,\tau$.  Using the conformal weights $h$, we can identify the corresponding anyon of the primary field $\Phi_{(1,2)}$ with $h=\frac{1}{10}$ as $\psi \boxtimes \bar{\tau}$, and $\Phi_{(1,3)}$ as $1\boxtimes \bar{\tau}$.  Because $\psi$ is a fermion, the monodromy representations from braiding $\psi \boxtimes \bar{\tau}$ are equivalent to the braidings of $\bar{\tau}$ up to phases.

The Fibonacci anyon $\tau$ is universal for quantum computation by braiding alone \cite{bib:2}.  Since the monodromy representations of braiding its complex conjugate $\bar{\tau}$ is simply the complex conjugate, $\bar{\tau}$ is also universal, thus it follows that $\psi \boxtimes \bar{\tau}$ is universal for quantum computation by braiding alone as well.

%\subsubsection{Monodromy representations}

From our identification above, we can compute the monodromy representations of conformal blocks from braid group representations of Fibonacci anyons up to overall phase factors.

In the graphical calculus of anyon models, conformal blocks are represented as fusion trees labeled by anyons.  The labeled trees in earlier sections can be directly translated into the fusion trees in anyon language.

To compute braid group representations, we select a basis of fusion trees and begin with the representatives for the braid group generators.  For example, we may compute the representation of $B_4$, the braid group on four strands, in the following orthonormal basis of Fibonacci fusion trees.
\[
    \bigbrace{
        \stikz{
            \node (X) at (0,-1.25) {};
            \node[above] at (0,-0.75) {$\tau$};
            \node[above] at (0.5,-0.75) {$\tau$};
            \node[above] at (1,-0.75) {$\tau$};
            \node[above] at (1.5,-0.75) {$\tau$};
            \draw (0,-0.75) -- (0.25,-1);
            \draw (0.5,-0.75) -- (0.25,-1);
            \draw (1,-0.75) -- (1.25,-1);
            \draw (1.5,-0.75) -- (1.25,-1);
            \draw (0.25,-1) -- (0.75,-1.5);
            \draw (1.25,-1) -- (0.75,-1.5);
            \node[left] at (0.5,-1.3) {\small $1$};
            \node[right] at (1,-1.3) {$\tau$};
            \draw (0.75,-1.5) -- (0.75,-1.75);
            \node[below] at (0.75,-1.75) {$\tau$};
        },
        \stikz{
            \node (X) at (0,-1.25) {};
            \node[above] at (0,-0.75) {$\tau$};
            \node[above] at (0.5,-0.75) {$\tau$};
            \node[above] at (1,-0.75) {$\tau$};
            \node[above] at (1.5,-0.75) {$\tau$};
            \draw (0,-0.75) -- (0.25,-1);
            \draw (0.5,-0.75) -- (0.25,-1);
            \draw (1,-0.75) -- (1.25,-1);
            \draw (1.5,-0.75) -- (1.25,-1);
            \draw (0.25,-1) -- (0.75,-1.5);
            \draw (1.25,-1) -- (0.75,-1.5);
            \node[left] at (0.5,-1.3) {$\tau$};
            \node[right] at (1,-1.3) {\small $1$};
            \draw (0.75,-1.5) -- (0.75,-1.75);
            \node[below] at (0.75,-1.75) {$\tau$};
        },
        \stikz{
            \node (X) at (0,-1.25) {};
            \node[above] at (0,-0.75) {$\tau$};
            \node[above] at (0.5,-0.75) {$\tau$};
            \node[above] at (1,-0.75) {$\tau$};
            \node[above] at (1.5,-0.75) {$\tau$};
            \draw (0,-0.75) -- (0.25,-1);
            \draw (0.5,-0.75) -- (0.25,-1);
            \draw (1,-0.75) -- (1.25,-1);
            \draw (1.5,-0.75) -- (1.25,-1);
            \draw (0.25,-1) -- (0.75,-1.5);
            \draw (1.25,-1) -- (0.75,-1.5);
            \node[left] at (0.5,-1.3) {$\tau$};
            \node[right] at (1,-1.3) {$\tau$};
            \draw (0.75,-1.5) -- (0.75,-1.75);
            \node[below] at (0.75,-1.75) {$\tau$};
        }
    }
\]
Given a braid, we use the Fibonacci $F$-symbols and $R$-symbols to write the result of braiding each of these basis vectors as a linear combination of the basis itself.  Thus each braid is assigned a matrix which is the change of basis matrix from a braided basis to this unbraided one.  The two generators $\sigma_1,\sigma_3$ of $B_4$ require only a single $R$ move to untangle.  The computation for the generator $\sigma_2$ goes as follows.
{\allowdisplaybreaks
\begin{align*}
    \stikz{
        \node (X) at (0,-0.875) {};
        \pic[braid/number of strands=4,braid/anchor=1-1,braid/width=0.5cm,braid/height=0.5cm] {braid={s_2^{-1}}};
        \node[left] at (0,-0.75) {$\tau$};
        \node[right] at (0.5,-0.75) {$\tau$};
        \node[right] at (1,-0.75) {$\tau$};
        \node[right] at (1.5,-0.75) {$\tau$};
        \draw (0,-0.75) -- (0.25,-1);
        \draw (0.5,-0.75) -- (0.25,-1);
        \draw (1,-0.75) -- (1.25,-1);
        \draw (1.5,-0.75) -- (1.25,-1);
        \draw (0.25,-1) -- (0.75,-1.5);
        \draw (1.25,-1) -- (0.75,-1.5);
        \node[left] at (0.5,-1.3) {\small $1$};
        \node[right] at (1,-1.3) {$\tau$};
        \draw (0.75,-1.5) -- (0.75,-1.75);
        \node[right] at (0.75,-1.75) {$\tau$};
    }
    &=
        F_{\tau;\tau\tau}^{1\tau\tau}
        \stikz{
            \node (X) at (0,-0.875) {};
            \pic[braid/number of strands=4,braid/anchor=1-1,braid/width=0.5cm,braid/height=0.5cm] {braid={s_2^{-1}}};
            \node[left] at (0,-0.75) {$\tau$};
            \node[right] at (0.5,-0.75) {$\tau$};
            \node[right] at (1,-0.75) {$\tau$};
            \node[right] at (1.5,-0.75) {$\tau$};
            \draw (0,-0.75) -- (0.25,-1);
            \draw (0.5,-0.75) -- (0.25,-1);
            \draw (1,-0.75) -- (0.5,-1.25);
            \draw (1.5,-0.75) -- (1.25,-1);
            \draw (0.25,-1) -- (0.75,-1.5);
            \draw (1.25,-1) -- (0.75,-1.5);
            \node[left] at (0.375,-1.15375) {\small $1$};
            \node[left] at (0.625,-1.4375) {$\tau$};
            \draw (0.75,-1.5) -- (0.75,-1.75);
            \node[right] at (0.75,-1.75) {$\tau$};
        } \\
    =
        F_{\tau;11}^{\tau\tau\tau}
        &\stikz{
            \node (X) at (0,-0.875) {};
            \pic[braid/number of strands=4,braid/anchor=1-1,braid/width=0.5cm,braid/height=0.5cm] {braid={s_2^{-1}}};
            \node[left] at (0,-0.75) {$\tau$};
            \node[left] at (0.5,-0.75) {$\tau$};
            \node[right] at (1,-0.75) {$\tau$};
            \node[right] at (1.5,-0.75) {$\tau$};
            \draw (0,-0.75) -- (0.25,-1);
            \draw (0.5,-0.75) -- (0.75,-1);
            \draw (1,-0.75) -- (0.5,-1.25);
            \draw (1.5,-0.75) -- (1.25,-1);
            \draw (0.25,-1) -- (0.75,-1.5);
            \draw (1.25,-1) -- (0.75,-1.5);
            \node[right] at (0.625,-1.15375) {\small $1$};
            \node[left] at (0.625,-1.4375) {$\tau$};
            \draw (0.75,-1.5) -- (0.75,-1.75);
            \node[right] at (0.75,-1.75) {$\tau$};
        }
        +
        F_{\tau;\tau1}^{\tau\tau\tau}
        \stikz{
            \node (X) at (0,-0.875) {};
            \pic[braid/number of strands=4,braid/anchor=1-1,braid/width=0.5cm,braid/height=0.5cm] {braid={s_2^{-1}}};
            \node[left] at (0,-0.75) {$\tau$};
            \node[left] at (0.5,-0.75) {$\tau$};
            \node[right] at (1,-0.75) {$\tau$};
            \node[right] at (1.5,-0.75) {$\tau$};
            \draw (0,-0.75) -- (0.25,-1);
            \draw (0.5,-0.75) -- (0.75,-1);
            \draw (1,-0.75) -- (0.5,-1.25);
            \draw (1.5,-0.75) -- (1.25,-1);
            \draw (0.25,-1) -- (0.75,-1.5);
            \draw (1.25,-1) -- (0.75,-1.5);
            \node[right] at (0.625,-1.15375) {$\tau$};
            \node[left] at (0.625,-1.4375) {$\tau$};
            \draw (0.75,-1.5) -- (0.75,-1.75);
            \node[right] at (0.75,-1.75) {$\tau$};
        } \\
    =
        F_{\tau;11}^{\tau\tau\tau}
        &R_1^{\tau\tau}
        \stikz{
            \node (X) at (0,-1.25) {};
            \node[above] at (0,-0.75) {$\tau$};
            \node[above] at (0.5,-0.75) {$\tau$};
            \node[above] at (1,-0.75) {$\tau$};
            \node[above] at (1.5,-0.75) {$\tau$};
            \draw (0,-0.75) -- (0.25,-1);
            \draw (0.5,-0.75) -- (0.75,-1);
            \draw (1,-0.75) -- (0.5,-1.25);
            \draw (1.5,-0.75) -- (1.25,-1);
            \draw (0.25,-1) -- (0.75,-1.5);
            \draw (1.25,-1) -- (0.75,-1.5);
            \node[right] at (0.625,-1.15375) {\small $1$};
            \node[left] at (0.625,-1.4375) {$\tau$};
            \draw (0.75,-1.5) -- (0.75,-1.75);
            \node[below] at (0.75,-1.75) {$\tau$};
        }
        +
        F_{\tau;\tau1}^{\tau\tau\tau}
        R_{\tau}^{\tau\tau}
        \stikz{
            \node (X) at (0,-1.25) {};
            \node[above] at (0,-0.75) {$\tau$};
            \node[above] at (0.5,-0.75) {$\tau$};
            \node[above] at (1,-0.75) {$\tau$};
            \node[above] at (1.5,-0.75) {$\tau$};
            \draw (0,-0.75) -- (0.25,-1);
            \draw (0.5,-0.75) -- (0.75,-1);
            \draw (1,-0.75) -- (0.5,-1.25);
            \draw (1.5,-0.75) -- (1.25,-1);
            \draw (0.25,-1) -- (0.75,-1.5);
            \draw (1.25,-1) -- (0.75,-1.5);
            \node[right] at (0.625,-1.15375) {$\tau$};
            \node[left] at (0.625,-1.4375) {$\tau$};
            \draw (0.75,-1.5) -- (0.75,-1.75);
            \node[below] at (0.75,-1.75) {$\tau$};
        } \\
    &\hspace{-2.5cm}=
        \Big(
            F_{\tau;11}^{\tau\tau\tau}
            R_1^{\tau\tau}
            \pn{F_{\tau}^{\tau\tau\tau}}^{-1}_{11}
            +
            F_{\tau;\tau1}^{\tau\tau\tau}
            R_{\tau}^{\tau\tau}
            \pn{F_{\tau}^{\tau\tau\tau}}^{-1}_{1\tau}
        \Big)
        \stikz{
            \node (X) at (0,-1.25) {};
            \node[above] at (0,-0.75) {$\tau$};
            \node[above] at (0.5,-0.75) {$\tau$};
            \node[above] at (1,-0.75) {$\tau$};
            \node[above] at (1.5,-0.75) {$\tau$};
            \draw (0,-0.75) -- (0.25,-1);
            \draw (0.5,-0.75) -- (0.25,-1);
            \draw (1,-0.75) -- (1.25,-1);
            \draw (1.5,-0.75) -- (1.25,-1);
            \draw (0.25,-1) -- (0.75,-1.5);
            \draw (1.25,-1) -- (0.75,-1.5);
            \node[left] at (0.5,-1.3) {\small $1$};
            \node[right] at (1,-1.3) {$\tau$};
            \draw (0.75,-1.5) -- (0.75,-1.75);
            \node[below] at (0.75,-1.75) {$\tau$};
        } \\
        \qquad
        +
        \Big(
            F_{\tau;11}^{\tau\tau\tau}
            &R_1^{\tau\tau}
            \pn{F_{\tau}^{\tau\tau\tau}}^{-1}_{\tau1}
            +
            F_{\tau;\tau1}^{\tau\tau\tau}
            R_{\tau}^{\tau\tau}
            \pn{F_{\tau}^{\tau\tau\tau}}^{-1}_{\tau\tau}
        \Big) \\
        &\hspace{-1.75cm}\left(
            \pn{F_{\tau}^{\tau\tau\tau}}^{-1}_{1\tau}
            \stikz{
                \node (X) at (0,-1.25) {};
                \node[above] at (0,-0.75) {$\tau$};
                \node[above] at (0.5,-0.75) {$\tau$};
                \node[above] at (1,-0.75) {$\tau$};
                \node[above] at (1.5,-0.75) {$\tau$};
                \draw (0,-0.75) -- (0.25,-1);
                \draw (0.5,-0.75) -- (0.25,-1);
                \draw (1,-0.75) -- (1.25,-1);
                \draw (1.5,-0.75) -- (1.25,-1);
                \draw (0.25,-1) -- (0.75,-1.5);
                \draw (1.25,-1) -- (0.75,-1.5);
                \node[left] at (0.5,-1.3) {$\tau$};
                \node[right] at (1,-1.3) {\small $1$};
                \draw (0.75,-1.5) -- (0.75,-1.75);
                \node[below] at (0.75,-1.75) {$\tau$};
            }\right.
            \left.+
            \pn{F_{\tau}^{\tau\tau\tau}}^{-1}_{\tau\tau}
            \stikz{
                \node (X) at (0,-1.25) {};
                \node[above] at (0,-0.75) {$\tau$};
                \node[above] at (0.5,-0.75) {$\tau$};
                \node[above] at (1,-0.75) {$\tau$};
                \node[above] at (1.5,-0.75) {$\tau$};
                \draw (0,-0.75) -- (0.25,-1);
                \draw (0.5,-0.75) -- (0.25,-1);
                \draw (1,-0.75) -- (1.25,-1);
                \draw (1.5,-0.75) -- (1.25,-1);
                \draw (0.25,-1) -- (0.75,-1.5);
                \draw (1.25,-1) -- (0.75,-1.5);
                \node[left] at (0.5,-1.3) {$\tau$};
                \node[right] at (1,-1.3) {$\tau$};
                \draw (0.75,-1.5) -- (0.75,-1.75);
                \node[below] at (0.75,-1.75) {$\tau$};
            }
        \right) \\
    &=
        \mat{\gamma^{-1}e^{4\pi i/5} \\ \gamma^{-1}e^{-3\pi i/5} \\ -\gamma^{-3/2}e^{-3\pi i/5}}, \\
    \stikz{
        \node (X) at (0,-0.875) {};
        \pic[braid/number of strands=4,braid/anchor=1-1,braid/width=0.5cm,braid/height=0.5cm] {braid={s_2^{-1}}};
        \node[left] at (0,-0.75) {$\tau$};
        \node[right] at (0.5,-0.75) {$\tau$};
        \node[right] at (1,-0.75) {$\tau$};
        \node[right] at (1.5,-0.75) {$\tau$};
        \draw (0,-0.75) -- (0.25,-1);
        \draw (0.5,-0.75) -- (0.25,-1);
        \draw (1,-0.75) -- (1.25,-1);
        \draw (1.5,-0.75) -- (1.25,-1);
        \draw (0.25,-1) -- (0.75,-1.5);
        \draw (1.25,-1) -- (0.75,-1.5);
        \node[left] at (0.5,-1.3) {$\tau$};
        \node[right] at (1,-1.3) {$1$};
        \draw (0.75,-1.5) -- (0.75,-1.75);
        \node[right] at (0.75,-1.75) {$\tau$};
    }
    &=
        \mat{\gamma^{-1}e^{-3\pi i/5} \\ \gamma^{-1}e^{4\pi i/5} \\ -\gamma^{-3/2}e^{-3\pi i/5}}, \\
    \stikz{
        \node (X) at (0,-0.875) {};
        \pic[braid/number of strands=4,braid/anchor=1-1,braid/width=0.5cm,braid/height=0.5cm] {braid={s_2^{-1}}};
        \node[left] at (0,-0.75) {$\tau$};
        \node[right] at (0.5,-0.75) {$\tau$};
        \node[right] at (1,-0.75) {$\tau$};
        \node[right] at (1.5,-0.75) {$\tau$};
        \draw (0,-0.75) -- (0.25,-1);
        \draw (0.5,-0.75) -- (0.25,-1);
        \draw (1,-0.75) -- (1.25,-1);
        \draw (1.5,-0.75) -- (1.25,-1);
        \draw (0.25,-1) -- (0.75,-1.5);
        \draw (1.25,-1) -- (0.75,-1.5);
        \node[left] at (0.5,-1.3) {$\tau$};
        \node[right] at (1,-1.3) {$\tau$};
        \draw (0.75,-1.5) -- (0.75,-1.75);
        \node[right] at (0.75,-1.75) {$\tau$};
    }
    &=
        \mat{-\gamma^{-3/2}e^{-3\pi i/5} \\ -\gamma^{-3/2}e^{-3\pi i/5} \\ \gamma^{-1}e^{3\pi i/5} - \gamma^{-3}}.
\end{align*}}
Thus, the generator $\sigma_2$ has a representation
\be
    \mat{\gamma^{-1}e^{4\pi i/5} & \gamma^{-1}e^{-3\pi i/5} & -\gamma^{-3/2}e^{-3\pi i/5} \\ \gamma^{-1}e^{-3\pi i/5} & \gamma^{-1}e^{4\pi i/5} & -\gamma^{-3/2}e^{-3\pi i/5} \\ -\gamma^{-3/2}e^{-3\pi i/5} & -\gamma^{-3/2}e^{-3\pi i/5} & \gamma^{-1}e^{3\pi i/5} - \gamma^{-3}}.
\ee
in agreement with our earlier result \eqref{eq:III-23}.

\section{Fault-tolerant Quantum Computation} \label{sec:V}
For quantum computation, we need to map the conformal blocks to quantum states. Then the braiding matrices become quantum gates. The conformal blocks $\mathcal{F}_\mu (\bm{\eta})$ involving $2N$ primary fields $ \sigma = \Phi_{(1,2)}$ (Eq.\ \eqref{eq:II-9}) cannot be interpreted as wavefunctions because of their singularities. There are two types of singularities, branch cuts and poles. To define a wavefunction, we insert into the correlator $2M$ fields $\psi$ that obey abelian fusion rules at positions $\bm{z} = (z_1, \dots, z_{2M})$. These are the coordinates of the wavefunction, whereas $\bm{\eta} = (\eta_1, \dots, \eta_{2N})$ are treated as parameters. Thus, we eliminate branch cuts associated with $\bm{\eta}$. The correlator still has poles at $\bm{z}$. To eliminate them, we introduce the Jastrow factor $\mathcal{J}$ that has zeroes at the position of these poles canceling the remaining singularities. Therefore, we are led to consider the wavefunction
\be\label{eq:V-1}  \Psi_{\mu ; \bm{\eta}} (\bm{z}) \propto \mathcal{J} (\bm{\eta} ; \bm{z}) \mathcal{F}_\mu (\bm{\eta} ; \bm{z}) , \ee 
defined as a product of two chiral $2N+2M -$point amplitudes, one (to be specified) determining $\mathcal{J}$, and another one involving the $\sigma$ and $\psi$ fields,
\be\label{eq:42} \mathcal{F}_\mu^{(2N,2M)} (\bm{\eta} ; \bm{z}) = \braket{\sigma_1 \cdots \sigma_{2N} \psi_1 \cdots \psi_{2M}} \ , \ee
where $\sigma_j = \sigma (\eta_j)$ and $\psi_j = \psi (z_j)$.

Braiding matrices act as unitary transformations mixing the states $\Psi_{\mu ; \bm{\eta}}$, as long as all conformal blocks yield states in the degenerate vacuum of the system. This leads to universal quantum computation for $k>2$.

Since $\psi$ obeys Abelian fusion rules, the amplitudes $\braket{\sigma_1 \cdots \sigma_{2N} \psi_1 \cdots \psi_{2M}}$ and $\braket{\sigma_1 \cdots \sigma_{2N}}$ have the same number of conformal blocks. Diagrammatically, they are shown in Figure \ref{fig:VI-1} for four insertions of $\sigma$ ($N=2$). Moving the $\psi$ insertions to different positions does not affect the conformal blocks.

\begin{figure}[ht!]
    \centering
    \begin{tikzpicture}
  \begin{feynman}
    \vertex at (0,0) (f1) {\(\sigma_1\)};
    \vertex at (1,0) (a2);
    \vertex at (1,1) (f2) {\(\sigma_2\)};
    \vertex at (2,0) (a3);
    \vertex at (2,1) (f3) {\(\sigma_3\)};
    \vertex at (3,0) (a4);
    \vertex at (3,1) (f4) {\(\psi_1\)};
    \vertex at (4,0) (a5);
    \vertex at (4,1) (f5) {\(\psi_{2M}\)};
    \vertex at (5,0) (f6) {\(\sigma_4\)};
    \vertex at (3.5,0.4) (c1) {\( \cdots \)};
    \vertex at (1.5,-0.3) (c2) {\(  \varepsilon\)};
    \vertex at (2.5,-0.3) (c3){\(\sigma\)};
    \vertex at (3.5,-0.3) (c4){\(\sigma\)};
    
    \diagram* {
      (f1) --  (a2) --  (f2),
      (a2) -- (a3) -- (f3),
      (a3) -- (a4) -- (f4),
      (a4) -- (a5) -- (f5),
      (a5)  -- (f6)
    };
  \end{feynman}
\end{tikzpicture}
    
    \begin{tikzpicture}
  \begin{feynman}
    \vertex at (0,0) (f1) {\(\sigma_1\)};
    \vertex at (1,0) (a2);
    \vertex at (1,1) (f2) {\(\sigma_2\)};
    \vertex at (2,0) (a3);
    \vertex at (2,1) (f3) {\(\sigma_3\)};
    \vertex at (3,0) (a4);
    \vertex at (3,1) (f4) {\(\psi_1\)};
    \vertex at (4,0) (a5);
    \vertex at (4,1) (f5) {\(\psi_{2M}\)};
    \vertex at (5,0) (f6) {\(\sigma_4\)};
    \vertex at (3.5,0.4) (c1) {\( \cdots \)};
    \vertex at (1.5,-0.3) (c2) {\( \mathbb{I} \)};
    \vertex at (2.5,-0.3) (c3){\(\sigma\)};
    \vertex at (3.5,-0.3) (c4){\(\sigma\)};
    
    \diagram* {
      (f1) --  (a2) --  (f2),
      (a2) -- (a3) -- (f3),
      (a3) -- (a4) -- (f4),
      (a4) -- (a5) -- (f5),
      (a5)  -- (f6)
    };
  \end{feynman}
\end{tikzpicture}
 
    \caption{The two conformal blocks for $N=2$ and arbitrary $M$.}
    \label{fig:VI-1}
\end{figure}
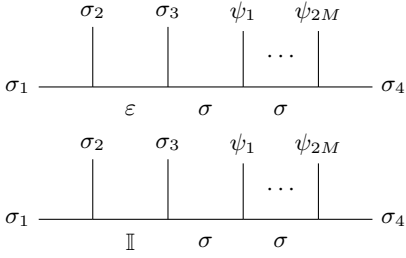

As shown in \cite{bib:20} using the plasma analogy, to construct a gapped state that will lead to fault-tolerant quantum computation, the dimension of $\psi$ must be less than $1$. In the critical Ising model $\mathcal{M}(4,3)$, this requirement is satisfied because $\psi$ can be chosen as the primary field $\Phi_{(1,3)}$ which has conformal dimension $h_{(1,3)} = \frac{1}{2} < 1$. The wavefunction \eqref{eq:V-1} is the MR wavefunction \cite{bib:14}. Unfortunately, braiding alone does not lead to universal quantum computation.

In the CFT minimal models $\mathcal{M}(k+2,k+1)$ with $k \geq 3$, we cannot identify $\psi$ with any of their primary fields, therefore we cannot construct a gapped wavefunction of the form \eqref{eq:V-1}.

A wavefunction of the form \eqref{eq:V-1} can be created from the coset CFT $SU(2)_1^{\otimes k} / SU(2)_{k}$ for all $k$, generalizing the MR wavefunction to which it reduces for $k=2$ (since the critical Ising model can be constructed from the coset CFT for $k=2$). As shown in Section \ref{sec:coamp}, the $\Phi_{(1,s)}$ primary fields with $1 \leq s \leq k+1$ in $\mathcal{M} (k+2,k+1)$ can be mapped onto primary fields in $SU(2)_1^{\otimes k} / SU(2)_{k}$ with the same fusion rules and conformal dimensions. Thus, correlators of the field $\sigma$ computed in $SU(2)_1^{\otimes k} / SU(2)_{k}$ are equivalent to those computed in the minimal model $\mathcal{M} (k+2,k+1)$. On the other hand, correlators involving the $\psi$ field must be computed in $SU(2)_1^{\otimes k} / SU(2)_{k}$. We define the $\psi$ field using
\be \label{eq:V-4}
\psi^{[ij]} = \chi^{[i]}_+ \chi^{[j]}_{-} +\chi^{[i]}_- \chi^{[j]}_{+}\ .
\ee
Evidently, \eqref{eq:V-4} does not lead to a unique definition of $\psi$ since we can consider any pair of $SU(2)_1$ factors in the coset to construct $\psi$. For desired results, one of the $SU(2)_1$ factors must be shared with the one in the definition of $\sigma$ (Eq.\ \eqref{eq:V-3a}). Since we identified $\sigma \equiv \sigma^{[k]}$, we will identify $\psi \equiv \psi^{[1k]}$. We also identify the conjugate field $\tilde{\sigma} = \Phi_{(k,k)} \equiv \sigma^{[1]}$, generalizing the critical Ising model result \cite{bib:41}. Then we obtain the fusion rules
\be \sigma \times \psi \sim \sigma \ , \ \ \psi\times\psi \sim \mathbb{I} \ , \ee 
for all $k$, same as in the critical Ising model. However, the fusion rule for $\sigma \times \sigma$ does not involve $\psi$ for $k>2$; it involves $\mathbb{I}, \varepsilon, \varepsilon' , \dots$, instead (whereas for $k=2$, $\psi = \varepsilon = \Phi_{(1,3)}$).

Using \eqref{eq:35}, we obtain its conformal weight $h_\psi = 2h_\chi = \frac{1}{2} <1$, satisfying the requirement for a gapped wavefunction \eqref{eq:V-1} that leads to fault-tolerant quantum computing.

%Moreover, since all the WZW primary fields involve spin $\frac{1}{2}$, we can omit the $j$ index.

Correlators of $\psi$ are constructed as products of $SU(2)_1$ correlators, similar to the factors of correlators of $\sigma$ (Eq.\ \eqref{eq:V-6}),
\be \label{eq:VI-6}
X^{[i]}_{\bm{m}}=
\braket{\chi^{[i]}_{m_1}(z_1) \cdots \chi^{[i]}_{m_{2M}}(z_{2M})} \ , \ \ i=1,\dots, k \ .
\ee
where $\bm{m} = (m_1,\dots, m_{2M})$. The $2M$-point $\psi$ correlator is found in terms of the Pfaffian, as in the case of the critical Ising model,
\be \label{eq:V-13}
\braket{\psi_1 \cdots \psi_{2M}} = \sum_{ \bm{m}} X^{[1]}_{\bm{m} } X^{[k]}_{\bar{\bm{m}}} = 2^{M} \text{Pf}\left({\frac{1}{z_i-z_j}}\right) \ ,
\ee
where $\bar{\bm{m}} = - \bm{m}$.

More generally, we construct the correlator $\mathcal{F}_{\mu}(\bm{\eta}; \bm{z})$  involving $2N$ $\sigma$ and $2M$ $\psi$ insertions (Eq.\ \eqref{eq:42}) as
\be \label{eq:V-14a}
\mathcal{F}_{\mu}(\bm{\eta};\bm{z}) \propto \sum_{\substack{n , m}} X_{\bm{m}}^{[1]} X_{\bar{\bm{m}},\bm{n}}^{[k]} Y_{\mu;\bm{n}} \ ,
\ee
in terms of three factors, similar to those in Eq.\ \eqref{eq:V-6},
\bea \label{eq:V-15}
X_{\bm{m}}&=&\braket{\chi_{m_1}(z_1) \cdots \chi_{m_{2M}}(z_{2M})} \nonumber \\
X_{\bar{\bm{m},\bm{n}}}&=&\braket{\chi_{n_1}(\eta_1) \cdots \chi_{n_{2N}}(\eta_{2N})\chi_{\bar{m}_1}(z_1) \cdots \chi_{\bar{m}_{2M}}(z_{2M}) } \nonumber \\
Y_{\mu ;\bm{n} } &=& \braket{\bar{\tau}_{n_1}(\eta_1) \cdots \bar{\tau}_{n_{2N}}(\eta_{2N})} \ .
\eea
Notice that the conformal block is only specified in the last factor because the fields in $SU(2)_1$ obey Abelian fusion rules. Also, we dropped the superscripts indicating which $SU(2)_1$ factor is used because it does not affect the calculation.

The remaining ingredient in the wavefunction \eqref{eq:V-1} is the Jastrow factor $\mathcal{J}$. We need a factor that cancels the poles of the conformal block $\mathcal{F}_\mu$ due to insertions of both $\sigma$ and $\psi$ fields. Following \cite{bib:20}, we define $\mathcal{J}$ as the correlation function for a free boson $\phi$,
\be \label{eq:V-16} \mathcal{J} = \langle \mathcal{V}_1 \dots \mathcal{V}_{2N} \mathcal{W}_1 \dots \mathcal{W}_{2M} \mathcal{Q} \rangle \ ,\ee
consisting of holomorphic vertices $\mathcal{V}$, $\mathcal{W}$, and a screening charge $\mathcal{Q}$, defined by
\be \label{eq:V-17}  \mathcal{V}_j = e^{i \frac{ 1}{2\sqrt{\Lambda}} \phi (\eta_j)} \ , \ \ \mathcal{W}_j = e^{i\sqrt{\Lambda} \phi (z_j)} \ , \nonumber \ee \be \label{eq:54} \mathcal{Q} = e^{-\frac{i}{\sqrt{\Lambda}} \int \frac{d^2w}{2\pi} \phi(w,\bar{w})}\ , \ee
Explicitly,
\bea
\mathcal{J}=&&\prod_{i<j}^{2M} z_{ij}^{\Lambda} \prod_{a< b}^{2N} \eta_{ab}^{\frac{1}{4\Lambda}} \prod_{a=1}^{2N} \prod_{i=1}^{2M} (\eta_a-z_i)^{\frac{1}{2}} \nonumber \\
&&\times  e^{-\frac{1}{4} \sum_{i=1}^{2M} |z_i|^2} e^{-\frac{1}{8\Lambda} \sum_{a=1}^{2N} |\eta_a|^2}
\eea
where $\Lambda$ is a positive integer that represents the inverse filling of FQHE \cite{bib:42}.

For $N=M=1$, we have a single conformal block. We obtain
\be \mathcal{F}^{(2,2)} = \mathcal{X}_{+-} Y_{+-} + \mathcal{X}_{-+} Y_{-+} \ee 
where
\be \mathcal{X}_{m_1m_2} = \sum_{m_3,m_4} X_{m_3m_4} X_{m_1m_2\bar{m}_3 \bar{m}_4} \ . \ee 
We have $Y_{+-} = Y_{-+} = \eta_{12}^{\frac{3}{2(k+2)}}$, and all $SU(2)_1$ correlators are easily computed. We obtain
\be \mathcal{F}^{(2,2)} = 2 \frac{ \eta_{12}^{\frac{1-k}{2k+4}}}{z_{12}} \frac{(\eta_1-z_1) (\eta_2-z_2) + (\eta_1-z_2) (\eta_2-z_1)}{ \sqrt{(\eta_1-z_1) (\eta_1-z_2) (\eta_2-z_1) (\eta_2-z_2)}} \ee
Notice that the exchange $\eta_1 \leftrightarrow \eta_2$ leads to the same factor as the one obtained for the propagator $\braket{\sigma_1\sigma_2}$.

The corresponding Jastrow factor is found to be
\bea \mathcal{J}^{(2,2)} &=&   \sqrt{(\eta_1-z_1) (\eta_1-z_2) (\eta_2-z_1) (\eta_2-z_2)} \nonumber \\
&& \times   z_{12}^{\Lambda}   \eta_{12}^{\frac{1}{4\Lambda}}
e^{-\frac{|z_1|^2+|z_2|^2}{4}  } e^{-\frac{   |\eta_1|^2+|\eta_2|^2}{8\Lambda}} \eea
and the wavefunction is
\be \Psi_{\eta_1,\eta_2} (z_1,z_2) \propto  \eta_{12}^{\frac{1}{4\Lambda} - \frac{k-1}{2(k+2)}}  z_{12}^{\Lambda -1 } \xi  e^{-\frac{|z_1|^2+|z_2|^2}{4}  } e^{-\frac{   |\eta_1|^2+|\eta_2|^2}{8\Lambda}} \ee 
where $\xi = (\eta_1-z_1) (\eta_2-z_2) + (\eta_1-z_2) (\eta_2-z_1)$ is a polynomial in $(z_1,z_2)$. Thus, $\Psi$ has no singularities in $(z_1,z_2)$.

For $N=2$ and $M=1$, we have two conformal blocks. Omitting overall normalization constants (\textit{cf}.\ with Eq.\ \eqref{eq:V-7}), after some algebra we obtain
\be
\mathcal{F}^{(4,2)}_\mu =  \frac{2}{z_{12}} \prod_{i=1}^2 \prod_{a=1}^4 (\eta_{a}-z_{i})^{-\frac{1}{2}} \Xi \ ,
\ee
where $\Xi$ is a polynomial in $(z_1,z_2)$,
\bea \Xi &=& \mathcal{X}_{+--+} Y_{\mu;+--+} + \mathcal{X}_{--++} Y_{\mu;--++} \nonumber\\ &&+ \mathcal{X}_{-+-+} Y_{\mu;-+-+} \ , \eea
and we defined
\bea \mathcal{X}_{+--+} &=&  \xi_{(14)(23)} X_{+--+} \nonumber\\
\mathcal{X}_{--++} &=&  \xi_{(12)(34)} X_{--++} \nonumber\\
\mathcal{X}_{-+-+} &=&  \xi_{(13)(24)} X_{-+-+} \ ,
\label{eq:62}\eea 
in terms of the polynomials
\be
\xi_{(ab)(cd)}=   (\eta_a-z_1) (\eta_b-z_1) (\eta_c-z_2) (\eta_d-z_2) + (z_1 \leftrightarrow z_2)  \ .
\ee
Notice that under the exchange $\eta_1 \leftrightarrow \eta_3$, we have $\mathcal{X}_{+--+} \leftrightarrow \mathcal{X}_{--++}$, and $\mathcal{X}_{-+-+}$ does not change, showing that $\mathcal{X}_{m_1m_2m_3m_4}$ (Eq.\ \eqref{eq:62}) have the same transformation properties as $X_{m_1m_2m_3m_4}$ (Eqs.\ \eqref{eq:V-8} and \eqref{eq:39}). Moreover, the correlators $Y_{\mu; m_1m_2m_3m_4}$ are the same as in the case of four-point amplitudes (Eq.\ \eqref{eq:42}). Therefore, the braiding rules for $\mathcal{F}_\mu^{(4,2)}$ are the same as in the absence of the field $\psi$ ($\mathcal{F}_\mu^{(4)}$, given by Eq.\ \eqref{eq:43}). The fact that braiding rules are not affected by $\psi$ is easily generalized to an arbitrary number of insertions of $\psi$ ($M > 1$).
 
The corresponding Jastrow factor is found to be
\bea \mathcal{J}^{(4,2)} &=& z_{12}^{\Lambda}  \prod_{a<b}^{4}\eta_{ab}^{\frac{1}{4\Lambda}}  \prod_{a=1}^4 \prod_{i=1}^2 (\eta_a-z_i)^{\frac{1}{2}} \nonumber \\
&& \times    
e^{-\frac{|z_1|^2+|z_2|^2}{4}  } e^{-\frac{|\eta_1|^2 +\cdots+|\eta_4|^2}{8\Lambda}} \eea
and the wavefunctions for the two conformal blocks are
\be \Psi_{\mu; \bm{\eta}} (z_1,z_2) \propto \prod_{a<b}^{4}\eta_{ab}^{\frac{1}{4\Lambda}} z_{12}^{\Lambda-1} \Xi\,      e^{-\frac{|z_1|^2+|z_2|^2}{4}  } e^{-\frac{|\eta_1|^2 +\cdots+|\eta_4|^2}{8\Lambda}}
\ee
They have no singularities for  $\Lambda \geq 1$.
In the case of the critical Ising model $(k=2)$, they reduce to the wavefunctions derived in Ref.\ \cite{bib:20}.

It is straightforward, albeit cumbersome, to generalize the above results to arbitrary $N,M$.

\section{Conclusion}\label{sec:con}
In this work, we developed a method to generalize the Moore-Read Pfaffian wavefunction \cite{bib:14} in a way that leads to fault-tolerant universal quantum computation. 

The MR wavefunction can be written in terms of correlators of the critical Ising model which is the minimal model CFT $\mathcal{M} (4,3)$. The correlators contain insertions of the field $\psi$ which has conformal weight $h = \frac{1}{2}$. Fault-tolerant quantum computing follows from the fact that the wavefunction is gapped due to the conformal weight of $\psi$ being less than one. Unfortunately, universal quantum computing cannot be achieved with braiding alone. Generalizations involving correlators of minimal models $\mathcal{M} (m+1,m)$ with $m>3$ lead to universal quantum computing, however, no field of conformal weight less than one has been identified, and therefore the resulting wavefunction is not gapped.

Instead of relying on minimal models, we constructed the wavefunction using conformal blocks of the coset $SU(2)_1^{\otimes k}/SU(2)_k$ which contain an Abelian primary field $\psi$ of conformal dimension less than one. We showed that the coset CFT $SU(2)_1^{\otimes k}/SU(2)_k$ contains a primary field of the same conformal weight, fusion rules and correlators as $\sigma \equiv \Phi_{(1,2)}$ in the minimal model for $m=k+1$. Additional, the coset CFT contains a field $\psi$ of conformal weight $h= \frac{1}{2}$ and Abelian fusion rules. These properties allowed us to use correlators of the coset CFT to generalize the MR wavefunction in a way that leads to fault-tolerant universal quantum computing.
For $k=2$ $(k=3)$ we recover the Ising (Fibonacci) anyons.

It would be interesting to find a system that will provide a physical realization of our wavefunction, similar to the realization of the MR wavefunction by the fractional quantum Hall effect at level $\nu=\frac{5}{2}$. In this respect, a comparison with the Read-Rezayi wavefunction \cite{bib:17}, which supports Fibonacci anyons (similar to our construction with $k=3$) and is realized by the FQHE at $\nu=\frac{12}{5}$ might be helpful. Work in this direction is in progress.

\acknowledgements
Research funded by ARO under grant W911NF-19-1-0397 and MURI contract W911NF-20-1-0082. G.S.\ is also partially supported by NSF grant OMA-1937008. Z.W.\ is also partially supported by NSF grants
FRG-1664351 and CCF-2006463.

\appendix 

\section{Details of calculations for the five-point amplitude}\label{app:A}
In the mimimal model $\mathcal{M}(k+2,k+1)$ with $k \geq 3$, the five-point correlator $\braket{\sigma_1 \sigma_2 \sigma_3 \sigma_4 \varepsilon_5}$ (Eq.\ \eqref{eq:III-16}) has 3 conformal blocks shown in Figure \ref{fig:III-3}.  The calculation of the exchange matrices involving the points $\eta_1$, $\eta_2$ and $\eta_3$ can be simplified by working in the $\eta_5\to \eta_4$ limit. To this end, we need to change bases so that $\sigma_4$ and $\varepsilon_5$ fuse together. A suitable change of basis is shown in Figure \ref{fig:A-1}. It involves four-point amplitudes that can be found explicitly.
\begin{figure}[ht]
    \centering
    \begin{tikzpicture}
    \begin{feynman}
    \vertex at (0,0) (a);
    \vertex at (-0.7,-0.7) (f1) {\(\varepsilon\)};
    \vertex at (-0.7,0.7) (f2) {\(\mathbb{I}\)};
    \vertex at (0.7,0) (b);
    \vertex at (0.35,0.3) (c) {\(\varepsilon\)};
    \vertex at (1.4,0.7) (f3) {\(\sigma\)};
       \vertex at (1.4,-0.7) (f4) {\(\sigma\)};
    \vertex at (1.8,0)  (d2) {\(=\)};
    
    \vertex at (2.4,0.7) (g1) {\(\mathbb{I}\)};
    \vertex at (4.1,0.7) (g2) {\(\sigma\)};
    \vertex at (3.25,0.35) (na);
    \vertex at (2.4,-0.7) (g3) {\(\varepsilon\)};
    \vertex at (4.1,-0.7) (g4) {\(\sigma\)};
    \vertex at (3.25,-0.35) (nb);
    \vertex at (3,0) (nc) {\(\sigma\)};
    \diagram* {
      (f1) --  (a) --  (f2),
      (a)  -- (b),
      (f3) -- (b) -- (f4),
      (g1) --  (na) --  (g2),
      (na) -- (nb),
      (g3) --  (nb) --  (g4),
    };
    \end{feynman}
    \end{tikzpicture}

    \begin{tikzpicture}
    \begin{feynman}
    \vertex at (0,0) (a);
    \vertex at (-0.7,-0.7) (f1) {\(\varepsilon\)};
    \vertex at (-0.7,0.7) (f2) {\(\varepsilon\)};
    \vertex at (0.7,0) (b);
    \vertex at (0.35,0.3) (c) {\(\mathbb{I}\)};
    \vertex at (1.4,0.7) (f3) {\(\sigma\)};
    \vertex at (1.4,-0.7) (f4) {\(\sigma\)};
    \vertex at (1.6,0)  (d2) {\(=\)};
    
    \vertex at (2.2,0) (nd) {\(D_{11}\)};
    \vertex at (2.4,0.7) (g1) {\(\varepsilon\)};
    \vertex at (4.1,0.7) (g2) {\(\sigma\)};
    \vertex at (3.25,0.35) (na);
    \vertex at (2.4,-0.7) (g3) {\(\varepsilon\)};
    \vertex at (4.1,-0.7) (g4) {\(\sigma\)};
    \vertex at (3.25,-0.35) (nb);
    \vertex at (3,0) (nc) {\(\sigma\)};
    
    \vertex at (4.3,0) (q) {\(+\)};
    \vertex at (4.8,0) (nq) {\(D_{12}\)};
    \vertex at (5,0.7) (ng1) {\(\varepsilon\)};
    \vertex at (6.7,0.7) (ng2) {\(\sigma\)};
    \vertex at (5.85,0.35) (nna);
    \vertex at (5,-0.7) (ng3) {\(\varepsilon\)};
    \vertex at (6.7,-0.7) (ng4) {\(\sigma\)};
    \vertex at (5.85,-0.35) (nnb);
    \vertex at (5.6,0) (nnc) {\(\varepsilon'\)};
    \diagram* {
      (f1) --  (a) --  (f2),
      (a)  -- (b),
      (f3) -- (b) -- (f4),
      (g1) --  (na) --  (g2),
      (na) -- (nb),
      (g3) --  (nb) --  (g4),
      (ng1) --  (nna) --  (ng2),
      (nna) -- (nnb),
      (ng3) --  (nnb) --  (ng4),
    };
    \end{feynman}
    \end{tikzpicture}
    
    \begin{tikzpicture}
    \begin{feynman}
    \vertex at (0,0) (a);
    \vertex at (-0.7,-0.7) (f1) {\(\varepsilon\)};
    \vertex at (-0.7,0.7) (f2) {\(\varepsilon\)};
    \vertex at (0.7,0) (b);
    \vertex at (0.35,0.3) (c) {\(\varepsilon\)};
    \vertex at (1.4,0.7) (f3) {\(\sigma\)};
    \vertex at (1.4,-0.7) (f4) {\(\sigma\)};
    \vertex at (1.6,0)  (d2) {\(=\)};
    
    \vertex at (2.2,0) (nd) {\(D_{21}\)};
    \vertex at (2.4,0.7) (g1) {\(\varepsilon\)};
    \vertex at (4.1,0.7) (g2) {\(\sigma\)};
    \vertex at (3.25,0.35) (na);
    \vertex at (2.4,-0.7) (g3) {\(\varepsilon\)};
    \vertex at (4.1,-0.7) (g4) {\(\sigma\)};
    \vertex at (3.25,-0.35) (nb);
    \vertex at (3,0) (nc) {\(\sigma\)};
    
    \vertex at (4.3,0) (q) {\(+\)};
    \vertex at (4.8,0) (nq) {\(D_{22}\)};
    \vertex at (5,0.7) (ng1) {\(\varepsilon\)};
    \vertex at (6.7,0.7) (ng2) {\(\sigma\)};
    \vertex at (5.85,0.35) (nna);
    \vertex at (5,-0.7) (ng3) {\(\varepsilon\)};
    \vertex at (6.7,-0.7) (ng4) {\(\sigma\)};
    \vertex at (5.85,-0.35) (nnb);
    \vertex at (5.6,0) (nnc) {\(\varepsilon'\)};
    \diagram* {
      (f1) --  (a) --  (f2),
      (a)  -- (b),
      (f3) -- (b) -- (f4),
      (g1) --  (na) --  (g2),
      (na) -- (nb),
      (g3) --  (nb) --  (g4),
      (ng1) --  (nna) --  (ng2),
      (nna) -- (nnb),
      (ng3) --  (nnb) --  (ng4),
    };
    \end{feynman}
    \end{tikzpicture}

    \caption{Basis change for the mixed four-point correlators. The LHS of the first, second and third line are denoted as $\mathcal{K}_1$, $\mathcal{K}_2$ and $\mathcal{K}_3$ respectively. On the RHS we have the correlators $\mathcal{K}'_1$, $\mathcal{K}'_2$ and $\mathcal{K}'_3$. }
    \label{fig:A-1}
\end{figure}
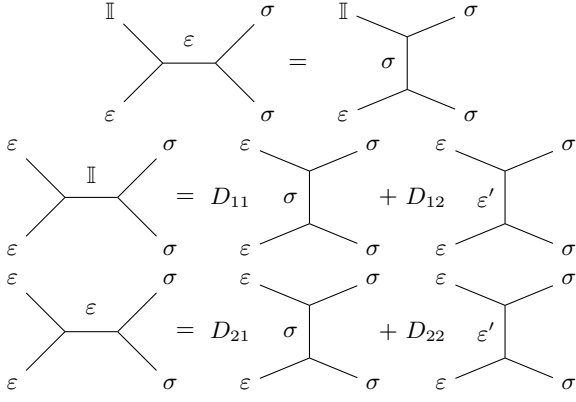
 The first correlator $\mathcal{K}_1$ transforms trivially. The other two correlators, $\mathcal{K}_2$ and $\mathcal{K}_3$, tranform to $\mathcal{K}_2'$ and $\mathcal{K}_3'$ via a matrix $D$, as shown. Working in the Coulomb gas formalism, after fixing three points, we obtain the four-point correlators in terms of Hypergeometric functions,
 \begin{widetext}
 \bea
\mathcal{K}_2(x)&=& x^{-\frac{2k}{k+2}} (1-x)^{\frac{k+1}{k+2}} \int_1^\infty dw w^{\frac{2k}{k+2}} (w-x)^{-\frac{2k+2}{k+2}} (w-1)^{-\frac{k+1}{k+2}} \nonumber\\
&=& \frac{ \Gamma \left(\frac{1}{k+2}\right)^2}{\Gamma \left(\frac{2}{k+2}\right)} (1-x)^{-\frac{k}{k+2}} x^{-\frac{2 k}{k+2}}   {}_2F_1\left(\frac{1}{k+2},-\frac{2 k}{k+2};\frac{2}{k+2};x\right) \ , \eea
\bea\mathcal{K}_3(x) &=& x^{-\frac{2k}{k+2}} (1-x)^{\frac{k+1}{k+2}}   \int_0^x dw w^{\frac{2k}{k+2}} (x-w)^{-\frac{2k+2}{k+2}} (1-w)^{-\frac{k+1}{k+2}} \nonumber\\
&=& \frac{ \Gamma \left(-\frac{k}{k+2}\right) \Gamma \left(\frac{3 k+2}{k+2}\right) }{\Gamma \left(\frac{2 (k+1)}{k+2}\right)} (1-x)^{\frac{k+1}{k+2}} x^{-\frac{k}{k+2}}  {}_2F_1\left(\frac{k+1}{k+2},\frac{3 k+2}{k+2};\frac{2 (k+1)}{k+2};x\right) \ ,
\eea
\bea
\mathcal{K}'_2(x) &=& x^{-\frac{2k}{k+2}} (1-x)^{\frac{k+1}{k+2}}  \int_x^1 dw w^{\frac{2k}{k+2}} (w-x)^{-\frac{2k+2}{k+2}} (1-w)^{-\frac{k+1}{k+2}} \nonumber\\
&=&  \frac{ \Gamma \left(\frac{1}{k+2}\right) \Gamma \left(-\frac{k}{k+2}\right)  }{\Gamma \left(\frac{1-k}{k+2}\right)} (1-x)^{-\frac{k}{k+2}} x^{-\frac{k}{k+2}}  {}_2F_1\left(-\frac{k}{k+2},\frac{k+1}{k+2};\frac{1-k}{k+2};1-x\right) \ , \eea
\bea\mathcal{K}'_3(x) &=& x^{-\frac{2k}{k+2}} (1-x)^{-\frac{k}{k+2}}  \int_x^1 dw w^{-\frac{2k+2}{k+2}} (w-x)^{\frac{2k}{k+2}} (1-w)^{-\frac{k+1}{k+2}} \nonumber\\
&=&  \frac{ \Gamma \left(\frac{1}{k+2}\right) \Gamma \left(\frac{3 k+2}{k+2}\right) }{\Gamma \left(\frac{3 (k+1)}{k+2}\right)} (1-x)^{\frac{k+1}{k+2}} x^{-\frac{2 k}{k+2}}  {}_2F_1\left(\frac{1}{k+2},\frac{2 (k+1)}{k+2};\frac{3 (k+1)}{k+2};1-x\right) \ .
\eea
 \end{widetext}
After some algebra we find
\be
D_{11}=  \frac{\sin\frac{2\pi}{k+2}}{\sin \frac{3\pi}{k+2} }  ,\ 
D_{12} =  \frac{\sin\frac{4\pi}{k+2}}{\sin \frac{3\pi}{k+2} }  ,\ 
D_{21} = D_{22} =  -\frac{\sin\frac{\pi}{k+2}}{\sin \frac{3\pi}{k+2} }  \ .
\ee
After applying this basis change to the five-point correlators depicted in Figure \ref{fig:III-3}, and using the OPE $ \sigma (\eta_4) \varepsilon(\eta_5) \sim \eta_{45}^{-\frac{k}{(k+2)}} \sigma(\eta_4) + \eta_{45}^{ \frac{k+1}{k+2}} \varepsilon' (\eta_4) $, we deduce in the limit $\eta_5\to \eta_4$
\bea\label{eq:A4}
\mathcal{F}_1^{(5)} &\approx &\eta_{45}^{-\frac{k}{k+2}} \mathcal{F}_1^{(4)} \ , \nonumber \\
\mathcal{F}_2^{(5)} &\approx & \eta_{45}^{-\frac{k}{k+2}} D_{11} \mathcal{F}_2^{(4)}  + \eta_{45}^{\frac{k+1}{k+2}} D_{12} \mathcal{F}_3^{(4)} \ ,\nonumber \\ 
\mathcal{F}_3^{(5)} &\approx & \eta_{45}^{-\frac{k}{k+2}} D_{21} \mathcal{F}_2^{(4)}  + \eta_{45}^{\frac{k+1}{k+2}} D_{22} \mathcal{F}_3^{(4)} 
\eea
where the four-point correlators $\mathcal{F}_1^{(4)} $, $\mathcal{F}_2^{(4)} $ are depicted in Figure \ref{fig:III-1} and  $\mathcal{F}_3^{(4)} $ is depicted in Figure \ref{fig:F0}.
\begin{figure}[H]
    \centering
    \begin{tikzpicture}
  \begin{feynman}
    \vertex at (0,0) (a);
    \vertex at (-1,-1) (f1) {\(\sigma_1\)};
    \vertex at (-1,1) (f2) {\(\sigma_2\)};
    \vertex at (1,0) (b);
    \vertex at (0.5,0.3) (c) {\(\varepsilon\)};
    \vertex at (2,1) (f3) {\(\tau_3\)};
    \vertex at (2,-1) (f4) {\(\varepsilon'_4\)};
    \vertex at (-2,0)  (d1) {\(\mathcal{F}_3^{(4)}\)};
    \vertex at (-1.5,0)  (d2) {\(=\)};
    \diagram* {
      (f1) --  (a) --  (f2),
      (a)  -- (b),
      (f3) -- (b) -- (f4)
    };
  \end{feynman}
\end{tikzpicture}
\caption{The four-point function $\braket{\sigma_1\sigma_2\sigma_3\varepsilon'_4}$}
\label{fig:F0}\end{figure}
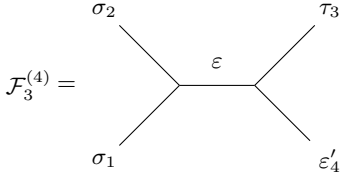

Since this correlator does not require screening charges, we readily deduce the algebraic expression 
\be
\mathcal{F}_3^{(4)} = \eta_{12}^{\frac{k}{2(k+2)}} (\eta_{13} \eta_{23})^{\frac{k+1}{2(k+2)}}  (\eta_{14} \eta_{24} \eta_{34})^{-\frac{3k+1}{2(k+2)}} \ .
\ee
Using the explicit expressions \eqref{eq:A4} involving four-point amplitudes, we easily obtain the exchange matrices of the five-point amplitudes depicted in Figure \ref{fig:III-3} corresponding to exchanges between the positions $\eta_1$, $\eta_2$, and $\eta_3$.

\section{Free-field representation of WZW models}\label{app:B}
A straightforward way to evaluate correlators in the $SU(2)_q$ WZW model is through the Wakimoto free-field representation in which the WZW model is expressed in terms of a free boson field $\varphi$ and a ghost system consisting of boson $\beta$ and $\gamma$ fields. The central charge $c$ of the theory and background charge $\alpha_0$ are given in terms of the level $q$ of the WZW model as
\be
c=3-12 \alpha_0^2= \frac{3q}{q+2}
\ee
The primary fields $\Phi_m^j(z)$  depend on two parameters taking the values $j=0,\frac{1}{2},\dots,\frac{q}{2}$ and $m=-j,\dots,j$. In the free-field representation,
\bea
\Phi_m^j(z) = \gamma^{j-m}(z) e^{-2ij \alpha_0 \varphi(z)} \ .
\eea
To compute correlators, we also need the conjugate fields $\tilde{\Phi}_m^j(z) $ and screening charges $Q_+$. The conjugate of the highest-weight field is \cite{bib:39}
\be
\tilde{\Phi}_j^j(z) = \beta^{2j-q-1}(z) e^{2i (j-q-1) \alpha_0 \varphi(z)}
\ee
The other fields can also be expressed in terms of boson fields, but we will not need explicit expressions.
There are two possible screening charges. For our purposes, we need one of them,
%We also need the screening charge given by 
\be
Q_+= \int dw \beta(w) e^{2i \alpha_0 \varphi(w)}
\ee
Concentrating on the primary fields with $j = \frac{1}{2}$, we simplify the notation by defining $\Phi_\pm \equiv \Phi_{\pm\frac{1}{2}}^{\frac{1}{2}}$.

In general, there are two conformal blocks in four-point correlators of fields with $j=\frac{1}{2}$, $X_{\mu, m_1m_2m_3m_4}$, where $\mu =1,2$ and $m_i=\pm$ ($i=1,2,3,4$). We obtain the non-vanishing amplitudes
%we can define the following 6 functions, based on the position of the screening charge and the $m=\pm \frac{1}{2}$ parameters. Of course for a nonzero amplitude we need an equal number of positive and negative m parameters.
\begin{widetext}
\bea\label{eq:B5}
X_{1,+--+} &=& \braket{\Phi_+(\eta_1) Q_+ \Phi_-(\eta_2) \Phi_-(\eta_3) \tilde\Phi_+(\eta_4)} \nonumber \\
&=& - \left[{x(1-x)}\right]^{\frac{1}{2(q+2)}} \int_0^x dw   \frac{\left[{w (x-w)(1-w)}\right]^{-\frac{1}{q+2}}}{x-w} - \left[{x(1-x)}\right]^{\frac{1}{2(q+2)}} \int_0^x dw   \frac{\left[{w (x-w)(1-w)}\right]^{-\frac{1}{q+2}}}{1-w}\nonumber \\
&=&   -\frac{ \Gamma \left(-\frac{1}{q+2}\right) \Gamma \left(\frac{q+1}{q+2}\right) }{\Gamma \left(\frac{q}{q+2}\right)} (1-x)^{\frac{1}{2 q+4}} x^{-\frac{3}{2 q+4}}  {}_2F_1\left(-\frac{1}{q+2},\frac{1}{q+2};\frac{q}{q+2};x\right) \ ,
\eea
\bea\label{eq:B6}
X_{1,--++} &=& \braket{\Phi_-(\eta_1) Q_+ \Phi_-(\eta_2) \Phi_+(\eta_3) \tilde\Phi_+(\eta_4)} \nonumber \\
&=& \left[{x(1-x)}\right]^{\frac{1}{2(q+2)}} \int_0^x dw   \frac{\left[{w (x-w)(1-w)}\right]^{-\frac{1}{q+2}}}{w} - \left[{x(1-x)}\right]^{\frac{1}{2(q+2)}} \int_0^x dw   \frac{\left[{w (x-w)(1-w)}\right]^{-\frac{1}{q+2}}}{x-w}\nonumber \\
&=& -\frac{ \Gamma \left(-\frac{1}{q+2}\right) \Gamma \left(\frac{q+1}{q+2}\right) }{q\Gamma \left(\frac{q}{q+2}\right)} (1-x)^{\frac{1}{2 q+4}} x^{\frac{2q+1}{2 q+4}} {}_2F_1\left(\frac{q+1}{q+2},\frac{q+3}{q+2};\frac{2 q+2}{q+2};x\right) \ ,
\eea
\bea\label{eq:B7}
X_{1,-+-+} &=& \braket{\Phi_-(\eta_1) Q_+ \Phi_+(\eta_2) \Phi_-(\eta_3) \tilde\Phi_+(\eta_4)} \nonumber \\
&=& \left[{x(1-x)}\right]^{\frac{1}{2(q+2)}} \int_0^x dw   \frac{\left[{w (x-w)(1-w)}\right]^{-\frac{1}{q+2}}}{w} - \left[{x(1-x)}\right]^{\frac{1}{2(q+2)}} \int_0^x dw   \frac{\left[{w (x-w)(1-w)}\right]^{-\frac{1}{q+2}}}{1-w}\nonumber \\
&=&   \frac{ \Gamma \left(-\frac{1}{q+2}\right) \Gamma \left(\frac{q+1}{q+2}\right) }{\Gamma \left(\frac{q}{q+2}\right)} (1-x)^{\frac{1}{2 q+4}} x^{-\frac{3}{2 q+4}}  {}_2F_1\left(\frac{1}{q+2},\frac{q+1}{q+2};\frac{q}{q+2};x\right) \ ,
\eea
\bea\label{eq:B8}
X_{2,+--+} &=& \braket{\Phi_+(\eta_1)  \Phi_-(\eta_2) \Phi_-(\eta_3)  Q_+ \tilde\Phi_+(\eta_4)} \nonumber \\
&=& \left[{x(1-x)}\right]^{\frac{1}{2(q+2)}} \int_1^\infty dw \frac{\left[{w (w-x) (w-1)}\right]^{-\frac{1}{q+2}} }{w-x} + \left[{x(1-x)}\right]^{\frac{1}{2(q+2)}} \int_1^\infty dw \frac{\left[{w (w-x) (w-1)}\right]^{-\frac{1}{q+2}} }{w-1}\nonumber \\
&=&  \frac{ \Gamma \left(-\frac{1}{q+2}\right) \Gamma \left(\frac{3}{q+2}\right) }{2 \Gamma \left(\frac{2}{q+2}\right)} (1-x)^{\frac{1}{2 q+4}} x^{\frac{1}{2 q+4}}  {}_2F_1\left(\frac{1}{q+2},\frac{3}{q+2};\frac{q+4}{q+2};x\right) \ ,
\eea
\bea\label{eq:B9}
X_{2,--++} &=& \braket{\Phi_-(\eta_1)  \Phi_-(\eta_2) \Phi_+(\eta_3)  Q_+ \tilde\Phi_+(\eta_4)} \nonumber \\
&=& \left[{x(1-x)}\right]^{\frac{1}{2(q+2)}} \int_1^\infty dw \frac{\left[{w (w-x) (w-1)}\right]^{-\frac{1}{q+2}} }{w} + \left[{x(1-x)}\right]^{\frac{1}{2(q+2)}} \int_1^\infty dw \frac{\left[{w (w-x) (w-1)}\right]^{-\frac{1}{q+2}} }{w-x}\nonumber \\
&=&   -\frac{\Gamma \left(-\frac{1}{q+2}\right) \Gamma \left(\frac{3}{q+2}\right)}{\Gamma \left(\frac{2}{q+2}\right)}  (1-x)^{\frac{1}{2 q+4}} x^{\frac{1}{2 q+4}} {}_2F_1\left(\frac{1}{q+2},\frac{3}{q+2};\frac{2}{q+2};x\right) \ , 
\eea
\bea\label{eq:B10}
X_{2,-+-+} &=& \braket{\Phi_-(\eta_1)  \Phi_+(\eta_2) \Phi_-(\eta_3)  Q_+ \tilde\Phi_+(\eta_4)} \nonumber \\
&=& \left[{x(1-x)}\right]^{\frac{1}{2(q+2)}} \int_1^\infty dw \frac{\left[{w (w-x) (w-1)}\right]^{-\frac{1}{q+2}} }{w} + \left[{x(1-x)}\right]^{\frac{1}{2(q+2)}} \int_1^\infty dw \frac{\left[{w (w-x) (w-1)}\right]^{-\frac{1}{q+2}} }{w-1}\nonumber \\
&=& \frac{\Gamma \left(-\frac{1}{q+2}\right) \Gamma \left(\frac{3}{q+2}\right)   }{2 \Gamma \left(\frac{2}{q+2}\right)} (1-x)^{-\frac{3}{2 (q+2)}} x^{\frac{1}{2 q+4}}  {}_2F_1\left(\frac{1}{q+2},\frac{q+1}{q+2};\frac{q+4}{q+2};x\right) \ .
\eea
\end{widetext}
Notice that only two of these functions are independent for each conformal block, because of the constraints
\be
X_{\mu,+--+} +X_{\mu,--++} +X_{\mu, -+-+} = 0 \ , \ \ \mu = 1,2 \ .
\ee
We obtain the expressions \eqref{eq:V-8} from the corresponding $X_{1,m_1m_2m_3m_4}$ for $q=1$. The second conformal block does not contribute in $SU(2)_1$ \cite{bib:25,bib:28}.
\end{document}